\documentclass{sig-alternate-2013-withpagenumbers}

\usepackage{mathtools,dsfont}
\usepackage{color,verbatim,enumitem,psfrag,bbm,comment}
\usepackage{algorithm,algorithmic}
\usepackage{appendix}

\usepackage{url}
\usepackage[colorlinks,
            linkcolor=blue,
            citecolor=blue,
            urlcolor=magenta
]{hyperref}

\def\EE{{\mathbb{E}}}\def\PP{{\mathbb{P}}}

\def\setN{\mathbb{N}}

\newcommand{\fall}{\,\forall\,}

\newtheorem{corollary}{Corollary}
\newtheorem{proposition}{Proposition}
\newtheorem{lemma}{Lemma}

\newtheorem{theorem}{Theorem} 
\newtheorem{definition}{Definition}

\newfont{\mycrnotice}{ptmr8t at 7pt}
\newfont{\myconfname}{ptmri8t at 7pt}
\let\crnotice\mycrnotice%
\let\confname\myconfname%

\makeatletter
\def\@copyrightspace{\relax}
\makeatother

\begin{document}
\title{The Behavior of Epidemics under Bounded Susceptibility
}

\numberofauthors{3} 

 \author{
 \alignauthor
 Subhashini Krishnasamy \\
        \affaddr{Department of ECE}\\
        \affaddr{UT Austin}\\
        \email{subhashini.kb@utexas.edu}
 \alignauthor
 Siddhartha Banerjee \\
        \affaddr{Department of MS\&E}\\
        \affaddr{Stanford University}\\
        \email{sidb@stanford.edu}
 \alignauthor Sanjay Shakkottai \\
        \affaddr{Department of ECE}\\
        \affaddr{UT Austin}\\
        \email{shakkott@austin.utexas.edu}
 }

\maketitle

\begin{abstract}
We investigate the sensitivity of epidemic behavior to a \emph{bounded susceptibility} constraint -- susceptible nodes are infected by their neighbors via the regular SI/SIS dynamics, but subject to a \emph{cap on the infection rate}. Such a constraint is motivated by modern social networks, wherein messages are broadcast to all neighbors, but attention spans are limited. Bounded susceptibility also arises in distributed computing applications with download bandwidth constraints, and in human epidemics under quarantine policies.

Network epidemics have been extensively studied in literature; prior work characterizes the graph structures required to ensure fast spreading under the SI dynamics, and long lifetime under the SIS dynamics. In particular, these conditions turn out to be meaningful for two classes of networks of practical relevance --  dense, uniform (i.e., \emph{clique-like}) graphs, and sparse, structured (i.e., \emph{star-like}) graphs. We show that bounded susceptibility has a surprising impact on epidemic behavior in these graph families. For the SI dynamics, bounded susceptibility has no effect on star-like networks, but dramatically alters the spreading time in clique-like networks. In contrast, for the SIS dynamics, clique-like networks are unaffected, but star-like networks exhibit a sharp change in extinction times under bounded susceptibility. 

Our findings are useful for the design of disease-resistant networks and infrastructure networks. More generally, they show that results for existing epidemic models are sensitive to modeling assumptions in non-intuitive ways, and suggest caution in directly using these as guidelines for real systems.
\end{abstract}

\keywords{Network Epidemics; SI/SIS Dynamics; Non-linear Dynamics; Bounded Susceptibility}

\section{Introduction}
\label{sec:intro}

Epidemic processes -- stochastic models for spreading in graphs -- play an important role in a variety of domains ranging from social networks (opinion spread), epidemiology (disease spread) and distributed computing (information spread). Though there are several models in literature for such processes, a characteristic feature of these is the spread of the process from `infected' to `non-infected' nodes along the edges of the network. In this work, we investigate the behavior of such processes under the assumption that \emph{the maximum susceptibility of a node to infection from its neighbors is bounded}. This constraint seems natural (and we present several settings where it is so) -- however we show that imposing it brings to light some surprising behavior of epidemic processes in different
networks.

In this paper, we focus on two major classes of epidemic processes: the $SI$ process where nodes never recover from infection, and the $SIS$ process where nodes recover, but can get reinfected. The $SI$ or `Susceptible-Infected' dynamics
\cite{Kesten03,BenPerTree,gopban11,BhamidiGNP} is the most common
model for one-way dissemination. Nodes exist in $2$ states --
`susceptible' ($S$) and `infected' ($I$). Infected nodes never recover, and try to infect each neighboring susceptible node at a rate $\beta$ (after a random delay, drawn from an i.i.d. $Exponential(\beta)$ distribution). We refer to $\beta$ as the \emph{virulence} of the infection process. The $SI$ dynamics terminate when all nodes are infected, and the corresponding time interval is called the \emph{spreading time}. In particular, we are interested in networks with small spreading time (ideally $O(1)$ in the number of nodes).

A small modification of this model gives the popular $SIS$ or
`Susceptible-Infected-Susceptible' dynamics
\cite{Liggett99,kepwhite91:viruses,satves02:scalefree,massganesh05,bbcs05,Banerjee2014a,BorgsAntidote}
for epidemic processes. This is identical to the $SI$ dynamics, except that now an infected node can revert to being susceptible at a (normalised) rate of `1'.
The absorbing state for the $SIS$ epidemic is when all nodes become susceptible; the time taken to reach this state is called the \emph{extinction time}. We are interested in networks where this extinction time is exponentially large (i.e., $\Omega(e^{poly(n)})$, where $n$ is the number of nodes in the network).

A key assumption in both models is that \emph{the infection rate for a susceptible node scales linearly in the size of the infected neighborhood}; formally, a susceptible node $j$ transitions to being infected at a rate $\beta|I_j|$, where $I_j$ is the set of infected neighbors of $j$. This assumption makes the models conceptually simple and tractable. However, it also means that in these models, high-degree nodes may get infected at unrealistically high rates. In several real-life settings, this is unreasonable; for example, in social networks, users broadcast content to all their neighbors, but attention-span constraints limit consumption of content to only a few neighbors. Similarly, in sensor networks using epidemic protocols, nodes can broadcast easily, but have download bandwidth constraints, and so can listen to only some of their neighbors. Refer to Section \ref{ssec:intro-appl} for more examples.

To modify epidemic dynamics to fit such settings, we need to assume that nodes have \emph{bounded susceptibility}, i.e., a maximum rate at which they can get infected, \emph{which is independent of network size}. We can capture this via the notion of an \emph{infection-profile}: a susceptible node $j$, we assume, is infected at a rate $\Phi(\beta|I_j|)$, where $\Phi(\cdot)$ is a \emph{bounded} function. Now, in order to use epidemic models to understand real-life processes, one hopes that the behavior of these epidemics is somewhat insensitive to the infection-profile $\Phi$. This motivates the central question we address in this paper:
\begin{center}
\em How robust is $SI$/$SIS$ epidemic behavior under an additional bounded susceptibility constraint?
\end{center}

The answer turns out to depend on the dynamics and network structure in a surprisingly subtle way. In particular, bounded susceptibility exposes a sharp dichotomy between two classes of networks of practical relevance:
\begin{itemize}[nolistsep,noitemsep]
\item \emph{Star-like networks,} i.e., sparse, structured graphs, with a small number of edges, mostly concentrated among a dense core, with all other nodes at a short distance from the core. These are good models for \emph{engineered networks}, which are designed to be highly connected while using only a small number of edges.
\item \emph{Clique-like networks,} i.e., dense, uniform graphs, with edges spread evenly among nodes. A good example is the connected-regime Erd\"{o}s-R\'{e}nyi graph (i.e., $G(n,p)$ with $p=\Omega(\log n/n)$). These are good models for \emph{organically-formed networks}.
\end{itemize} 
In prior work, the distinction between the two classes is suggested, but not explicitly made; we formalize this classification in Section \ref{ssec:model}. It turns out that this is crucial to our study of bounded susceptibility, as follows:

For a fixed virulence $\beta$, the spreading time in the $SI$ dynamics on star-like networks is small ($O(\log n)$, under fairly weak conditions on the non-core subgraph); in clique-like networks, the spreading time is $O(1)$\footnote{Special cases of this result can be found in \cite{BenPerTree,BhamidiGNP}}. In case of the $SIS$ epidemic, both clique-like networks and star-like networks experience exponential extinction times \cite{massganesh05}. Thus, in a sense, \emph{both star-like and clique-like networks are well suited for epidemic processes} -- this behavior, however, \emph{no longer holds under bounded susceptibility}.

\subsection{An Overview of our Results}
\label{ssec:contrib}

We refer to our modified epidemic models, incorporating bounded susceptibility, as the $SI_{bs}$ and $SIS_{bs}$ dynamics. As mentioned, we assume that in both models, a susceptible node $j$ gets infected at a rate $\Phi(\beta|I_j|)$, where $I_j$ are the infected neighbors of $j$; in the $SI_{bs}$ model, infected nodes remain so forever, while in the $SIS_{bs}$ model, they recover to become susceptible at a rate $1$. We assume that $\Phi(0)=0$, and $\Phi(\cdot)$ is concave, non-decreasing, and bounded above by some $\Phi_{\max}=\Theta(1)$.

\noindent\textbf{$SI_{bs}$ Dynamics:} In the $SI_{bs}$ model, we are interested in finding the (expected) spreading time, under the assumption that both $\beta$ and $\Phi_{\max}$ are $\Theta(1)$. In this context, we show:
\begin{itemize}[nolistsep,noitemsep]
\item For clique-like networks,  the spreading time switches from $O(1)$ under $SI$ dynamics to $\Omega(\log n)$ under the $SI_{bs}$ dynamics.
In particular, for a $G(n,p)$ graph in the connected regime (i.e., $p\sim\omega\left(\log n/n\right)$), the spreading time switches from $o(1)$ to $\Omega(\log n)$ under the $SI_{bs}$ dynamics.
\item For star-like networks, the spreading time is $\Theta(\log n)$ under both $SI$ and $SI_{bs}$ dynamics (under some assumptions on the network structure). Thus, imposing bounded susceptibility does not significantly alter the epidemic behavior.
\end{itemize}

Thus, the $SI$ model is highly sensitive to bounded susceptibility in clique-like networks -- note that the spreading time has switched from vanishing with network size to actually scaling with network size. On the other hand, it is insensitive in star-like networks.

A moment's thought suggests that this behavior is indeed reasonable -- we have greatly reduced the (potential) susceptibility of high-degree nodes, which should hinder the spread of the epidemic. Furthermore, since high-degree nodes are more affected, it seems natural that bounded susceptibility should adversely affect dense graphs much more than sparse graphs. Surprisingly, the latter intuition is completely reversed in case of the $SIS_{bs}$ model. 

\noindent\textbf{$SIS_{bs}$ Dynamics:} Recall that in the $SIS_{bs}$ model, we want to find the (expected) extinction time under $\beta,\Phi_{\max}=\Theta(1)$ and recovery rate $1$.
In this context, we show the following:
\begin{itemize}[nolistsep,noitemsep]
\item For clique-like networks, the extinction time is exponential (i.e., $\exp(\Omega(poly(n)))$) both with and without the bounded susceptibility assumption. In fact, this robustness property holds \emph{for any network which contains a polynomial-sized subgraph which resembles a random $G(n,p)$ graph in the connected regime} (i.e., with $p\sim\omega\left(\log n/n\right)$).
\item However for star-like networks, we now observe a sharp transition in extinction time. While star-like networks experience epidemics with exponential lifetimes under the SIS process \cite{massganesh05,bbcs05}, we prove that this is not true in the case of $SIS_{bs}$ infections. In particular, for any network which can be decomposed into a small \emph{core}, plus an outer sparse subgraph, \emph{the $SIS_{bs}$ epidemic lifetime is sub-exponential for any virulence $\beta=\Theta(1)$}.
\end{itemize}

This reversal of results may seem paradoxical -- there is however an intuitive explanation for this. Note that the absorbing states in the $SI$ and $SIS$ epidemics are exactly the opposite (all-infected vs. all-susceptible). Moreover, in the $SI$ model, we want the epidemic to move rapidly towards the all-infected state, while in the $SIS$, we want to deter it from reaching the all-susceptible state. Rapid transition to the all-infected state in the $SI$ dynamics is aided by having many high-degree nodes -- this is hampered by imposing bounded susceptibility. On the other hand in the $SIS$ dynamics, high degree nodes impede transition to the all-susceptible state -- having many such nodes helps distribute this effect, but too few high-degree nodes means that these are now critical bottlenecks. Under bounded susceptibility, these bottleneck nodes are the most affected, resulting in a sharp transition in the epidemic behavior.

To make this intuition more transparent, in Section \ref{ssec:example} we demonstrate the effect of the various dynamics in the clique and star networks, which are in a sense the extreme-case scenarios. Our main contribution in this work is to show how these phenomena extend to large classes of networks, occupying the space between these two extremes. In the process, we develop some structural lemmas relating the dynamics of epidemics to network properties, which may be of independent interest.

\paragraph*{Technical Contributions} 
Most of our arguments involve standard Markov chain machinery, along with careful use of measure concentration. Some of the ideas we use are of fairly recent origin, in particular, the sharp expressions for extinction time in SIS models via embedding in an ergodic Markov chain (Lemma \ref{lemma:absorption-time}, which we adapt from Lemma $8$ in \cite{Banerjee2014a}), and characterization of the $SI$ spreading time in clique-like graphs, which is similar to results for first-passage percolation in random graphs \cite{BhamidiGNP}. However, unlike \cite{BhamidiGNP}, our results (Theorems \ref{thm:SI-upperbnd},\ref{thm:SI-starlike}, \ref{thm:SIS-clique} and \ref{thm:SISbs-starlike}) do not focus on a particular graph or generative model, but rather, relate the epidemic behavior to structural properties of the network. We believe these lemmas could be useful in studying other related settings. Finally, our results follow from elementary arguments, and hold for finite networks, with explicit error bounds in terms of network size.

We note that our classification of networks into clique-like and star-like does not cover all networks. In particular, there are networks which are a mix of the two, and where our results may not be tight; sharper techniques could perhaps be used to bridge this gap. Furthermore, our techniques do not immediately extend to other epidemic models, in particular, the SIR dynamics (see \cite{DraiefMass,draief08} for details) -- this is an interesting avenue for future exploration.

\subsection{Bounded Susceptibility: Applications in\\ Real-World Settings}
\label{ssec:intro-appl}

Before we formally state our results, we first discuss three important real-life settings which exhibit bounded susceptibility, and point out applications of our results therein.

\paragraph*{Social and Economic Networks} 
Perhaps the most compelling example of bounded susceptibility in network epidemics comes from online social networks. Sites such as Twitter and Facebook are built on a broadcast model -- users upload content to the website, intending it to be broadcast to all neighbors. However, users are known to consume content from only a small set of neighbors -- this is due to a combination of limited attention-spans and content filtering algorithms, and has been observed empirically in Twitter data \cite{Hodas12,Weng12}. Epidemic models have been used to understand why content `goes viral' on such media, however existing works do not account for bounded susceptibility.

Non-linear propagation models are also used in economics to model network externalities in the spread of opinions, technologies, etc. For example, the adoption of a new technology on a network is often subject to diminishing returns due to wearing out of novelty with increasing number of recommendations \cite{Leskovec07} -- this naturally suggests a concave infection-profile, and further, bounded susceptibility.

\paragraph*{Human-Disease Epidemiology}
Network epidemic models are widely used to study human epidemics, and also computer viruses. Here again, bounded susceptibility is often a key feature of real epidemics, due to spread of awareness and adoption of preventive measures. Non-linear infection models have been used to capture the effect of risk perception in the spread of epidemics \cite{Funk09,Wu12}, by combining the effect of infection spread and the protection measures taken by the susceptible individual. However these works focus on approximate and mean-field treatments for the problem. 

Our work provides both intuition and technical tools for studying the effects of bounded susceptibility on epidemics. One important application of our findings is in the design of vaccination and quarantining policies on networks -- such policies help reduce the susceptibility of nodes to infection, but usually nodes retain some susceptibility. A closely related design question is of distributing antidotes in a network \cite{BorgsAntidote} -- there, the epidemic is controlled by increasing the recovery rate of a node, whereas our control is via limiting the susceptibility of a node.

\paragraph*{Distributed Computing on Networks}
A third application of epidemic processes is as a primitive for designing low-complexity distributed algorithms. The $SI$ dynamics is the basis of the \emph{flooding algorithm}, which is used for broadcast messaging in networks, distributed database synchronization, etc \cite{Eugster04}. The $SIS$ dynamics is sometimes used to ensure persistent storage in distributed storage/sensing applications with disk failures \cite{Chakrabarti07}. Existing work focuses on using the standard $SI/SIS$ dynamics, or modifies them to impose \emph{bounded influence}, i.e., a bound on the upload bandwidth of a node. 

However many of the above settings, in reality, exhibit bounded susceptibility. Modern distributed/P2P systems  are based on a virtual network architecture, where nodes upload content to a central location, from where it is downloaded by its neighbors. Thus, broadcasting is easy, but downloading is bandwidth constrained. This is also true for sensor networks -- a node needs to broadcast a message once, but it is limited in how many neighbors it can listen to, due to a fixed download bandwidth. Thus, bounded susceptibility is natural in these settings.
We discuss all these applications in Section \ref{ssec:disc}, after presenting our results.

\subsection{Related Work}
\label{ssec:relwork}

Epidemic processes on networks have been studied across many disciplines; readers interested in more details regarding epidemics are referred to several excellent books on the subject \cite{Kesten03,DraiefMass,andersonmay92:diseasesbook,Durrett07}. Specifically, there is a vast literature to characterize spreading time in various contexts for SI processes \cite{Kesten03,BenPerTree,gopban11,BhamidiGNP}, and phase transitions/extinction time for SIS processes \cite{kepwhite91:viruses,satves02:scalefree,massganesh05,bbcs05,Banerjee2014a,BorgsAntidote}. Phase transitions for SIR processes are available in \cite{draief08}.

Non-linear epidemic models have been explored in the past in various contexts. A large body of work is devoted to settings with \emph{bounded influence}, where the total rate at which a node infects its neighbors is bounded \cite{dshah:gossipbook,sanghajek07:gossip}, but the infection-profile $\Phi(\cdot)$ is still linear. One popular model with non-linear $\Phi(\cdot)$ is the Bootstrap Percolation, and related variants \cite{janson12,lelarge09}. Several authors have also studied non-linear infection models at the population level using mean field approximations \cite{Funk09, Wu12} -- these works show the existence of epidemic thresholds under different non-linear models. However, to the best of our knowledge, ours is the first work that explicitly characterizes epidemic behavior under bounded susceptibility.

\section{Main Results and Discussion}
\label{sec:results}

We now present and discuss our main results. Our general theorems are presented in Section \ref{ssec:results}, followed by a discussion of their application in Section \ref{ssec:disc}; before that, in Section \ref{ssec:example}, we specialize our results for the clique and star networks, and give an outline of the proofs for these cases. The complete proofs are somewhat technical -- we outline the main proof ideas after each theorem, deferring complete proofs to Section \ref{sec:proofs}.

\subsection{System Model}
\label{ssec:model}

%

\subsubsection*{The $SI$/$SIS$ Dynamics}

The $SI, SIS$ dynamics models we use are standard in literature \cite{massganesh05,Banerjee2014a}. We consider a graph $G(V,E)$ with $n$ nodes (vertices). The various processes $(SI, SI_{bs}, SIS, SIS_{bs})$ evolve on this graph in continuous time. Associated with each node $i \in \{1, 2, \ldots, n\}$ is a random process $X_i(t) \in \{0, 1\}.$ Here, `0' corresponds to the susceptible state and `1' corresponds to the infected state (the entire state vector is denoted by $\mathbf{X}(t)$).
For the $SI$ and $SIS$ models, nodes in state `1' (infected nodes) attempt to infect a neighbor \emph{independently} at a rate $\beta$, i.e., after a i.i.d random interval drawn from an $Exponential(\beta)$ distribution. In the $SI$ dynamics, an infected node never recovers. In contrast in the $SIS$ dynamics, an infected node returns back to the susceptible state at rate $1$ (i.e., after an $Exponential(1)$ time).

We assume that the infection starts off at a single node, arbitrarily chosen. This is natural for the $SI$ epidemic since we want to study spreading-time. For the $SIS$ epidemic, our results generalize for any constant-sized initial infected set. For the $SI$ dynamics, the only absorbing state is the \emph{all-infected} state (i.e., $X_i = 1\,\fall i$), and the spreading-time is defined as $\tau_s\triangleq\inf_t\{X_i(t)=1\,\fall i\}$. In contrast, for the $SIS$ dynamics, the only absorbing state is the \emph{all-susceptible} state (i.e., $X_i = 0\,\fall i$), and the extinction-time is defined as $\tau_e\triangleq\inf_t\{X_i(t)=0\,\fall i\}$. We also assume throughout that $\beta$ is a constant, independent of $n$; however, our analysis extends for settings where $\beta$ is a function of $n$.


\subsubsection*{Infection-Profiles and Bounded Susceptibility}

A critical assumption in both $SI$ and $SIS$ models is that the infection rate for a susceptible node scales in proportion to the number of infected neighbors. In order to study the effect of non-linearity in the infection-rate, we introduce the notion of an infection-profile: we assume that a susceptible node $j$ is infected at a rate $\Phi(\beta|I_j|)$, where $I_j$ is the number of infected neighbors of node $j$ (at some time $t$; we suppress the dependence on $t$ for ease of notation). Thus, in both the $SI_{bs}$ and the $SIS_{bs}$ dynamics, we have:
\begin{align*}
 & X_i: 0 \rightarrow 1 \quad \text{at rate } \Phi\left(\beta \sum_{(i,j) \in E} X_j\right),
\end{align*}
and in the $SIS_{bs}$, we additionally have $X_i: 1 \rightarrow 0 \text{ at rate }1$.

$\Phi: \mathbb{R_+} \rightarrow \mathbb{R_+}$ represents the overall effect that the infected neighbors together have on a susceptible node. We assume that $\Phi(0) = 0$, and that $\Phi(\cdot)$ is non-decreasing, concave and bounded. Note that these assumptions imply that $\Phi_{max}:= \sup_{x\in \mathbb{R}}\Phi(x)$ exists, and $\Phi(\beta)>0\fall\beta>0$.


\subsubsection*{Clique-Like and Star-Like Networks}

We now formalize the two classes of networks, which we introduced in Section \ref{sec:intro}. Intuitively, clique-like networks are those which are dense (i.e., where $|E|=\omega(n)$) and have edges evenly distributed across the cuts of the network. One way to formalize this is via the expansion properties of the underlying graph. There are different ways to quantify the expansion of a graph. For example, the \emph{edge expansion} of a graph is characterized by the \textit{generalized isoperimetric constant}:
\begin{align*}
\eta_m^e(G) := \inf_{S \subset V, |S|\leq m}\frac{|\delta(S)|}{|S|}, \quad 0 < m \leq \lfloor n/2 \rfloor,
\end{align*}
where, for any set $A\subseteq V$, $\delta(A)=\{(i,j)\in E \,\text{ s.t. }\, i\in A,j\notin A\}$ are the edges in the cut defined by $A$. The case of $m = n/2$ is referred to as the \textit{isoperimetric/Cheeger constant}:
\begin{align*}
\eta(G) := \inf_{S \subset V, |S|\leq n/2}\frac{|\delta(S)|}{|S|}.
\end{align*} 
We can also define an alternate notion of isoperimetry based on \textit{vertex expansion}, as follows:
\begin{align*}
\eta^v_m(G) := \inf_{S \subset V, |S|\leq m}\frac{|\Gamma(S)|}{|S|}, \quad 0 < m \leq \lfloor n/2 \rfloor,
\end{align*}
where $\Gamma(S)$ is the \emph{neighborhood} of $S$, i.e., the set of nodes in $S^C$ with at least one neighbor in $S$. It is easy to check that for any fixed $m$, we have $\eta^e_m(G)\geq\eta^v_m(G)$. 
 
The isoperimetric constants characterize bottleneck sets, i.e., those containing the smallest fraction of the potential edges in the cut. For a clique, it is easy to check that $\eta(G)=\Theta(n)$, which is the highest possible. To show that epidemic properties hold over a large class of graphs, we want to admit graphs with smaller isoperimetric constants. This motivates the following characterizations for clique-like networks:
\begin{definition}[Clique-like Networks]
\label{def:clique-like}
We define a graph $G$ to be clique-like if it satisfies one of the following:
\begin{itemize}[nolistsep,noitemsep]
\item[(A)] If $\eta(G)=\Omega(\log n)$.
\item[(B)] If for some $m = n^\alpha,$ $\alpha \in (0,1)$, we have $\eta_m^v(G) = \omega(1)$.
\end{itemize}
\end{definition}
The two definitions arise from analyzing the two epidemic dynamics on graphs. Although not equivalent, they both characterize networks with good expansion properties. To demonstrate their wide applicability, we show that random (\emph{Erd\"{o}s-R\'{e}nyi}) graphs in the connected regime (i.e., $G(n,p)$ with $p\sim\Omega\left(\log n/n\right)$) are clique-like under both definitions.

On the other hand, star-like networks intuitively are those which can be partitioned into a \emph{dense core} and a \emph{sparse periphery}. Here, the natural notions of density turn out to be the average degree $d_{avg}(G)$, and the spectral radius $\rho(G)$, i.e., the largest eigenvalue of the adjacency matrix, which is closely related to the node degrees (in particular, $d_{avg}(G)\leq\rho(G)\leq d_{\max}(G)$. Formally, we write $V=V_c\cup V_p$, where $V_c$ denotes the core and $V_p$ the periphery. As for clique-like networks, we have two characterizations:
\begin{definition}[Star-like Networks]
\label{def:star-like}
We define a\\ graph $G$ to be star-like if it satisfies one of the following:
\begin{itemize}[nolistsep,noitemsep]
\item[(A)] If $|V_p|= \Omega(n)$ and the average degree of nodes in the periphery is $O(1).$ In addition, the diameter of the graph is $O(\log n).$
\item[(B)] If $|V_c|=O(\mbox{poly}\log n)$ and the subgraph $G(V_p)$ induced by the periphery has spectral radius $\rho(G(V_p))=o(1).$
\end{itemize}
\end{definition}
Again, the class of graphs that are star-like by the two definitions above are not equivalent but broadly include networks which have a sparse periphery and a small (possibly dense) core. We derive our results for SI dynamics based on definition $(A)$ of clique-like and star-like networks, while the results for SIS dynamics are based on definition $(B)$ for both the networks.

\subsection{Two Examples: The Clique and the Star}
\label{ssec:example}

Before presenting our full results, we first illustrate them by considering our two representative graphs -- the clique and the star. This comparison brings out the dichotomy between the two network classes. Further, it allows us to build intuition for our main results, since these special cases admit simple proofs. In all these results, we use standard notation ($O$, $\Theta$, $\Omega$, $o$ and $\omega$) to characterize the scaling behavior with the graph size, $n$.  We also use the notation $[n]\triangleq\{1,2,\ldots,n\}$, and $\mathcal{H}_n=\sum_{i=1}^n\frac{1}{i}$.

\subsubsection*{$SI-SI_{bs}$ Dynamics}

Consider first the $SI/SI_{bs}$ epidemic on $K_n$, the clique on $n$ nodes. We now show that the spreading-time is \emph{vanishingly small} under the $SI$ dynamics, while it \emph{scales with $n$} under the $SI_{bs}$ dynamics.

\begin{proposition}
\label{prop:si-clique}
For the $SI/SI_{bs}$ epidemic on $K_n$ starting at an arbitrary node,
\begin{enumerate}[nolistsep,noitemsep]
\item Under the SI dynamics, $\EE[\tau_s] = \Theta(\frac{\log n}{n})$.
\item Under the $SI_{bs}$ dynamics $\EE[\tau_s] = \Theta(\log n)$.
\end{enumerate}
\end{proposition}

\begin{proof}
The first part is known from previous results \cite{BhamidiGNP} and follows from standard Markov-chain arguments. By denoting $N(t)$ to be the  number of infected nodes at time $t,$ and analyzing the steady-state of the resulting (one-dimensional) CTMC, we immediately get the result.

Next, for the bounded-susceptibility spreading process, $SI_{bs}$, we again consider the process $\{N(t)\}$, i.e., the total number of infected nodes. $\{N(t)\}$ has transition rates as follows:
\begin{align*}
& z \rightarrow z+1 \quad \text{at rate } \Phi(\beta z)(n-z) \quad \fall 1 \leq z \leq n-1,
\end{align*}
and hence the expected spreading-time is given by:
\begin{align*}
\EE [\tau_s] &= \sum_{k=1}^{n-1} \frac{1}{\Phi(\beta k)(n-k)}\geq \sum_{k=1}^{n-1} \frac{1}{(n-k) \Phi_{max}} = \frac{\mathcal{H}_{n-1}}{\Phi_{max}},
\end{align*}
where the inequality follows since $\Phi(\beta k) \leq \Phi_{max}$. From the same equation, we also have:
\begin{align*}
\EE [\tau_s] & = \sum_{k=1}^{n-1} \frac{1}{\Phi(\beta k)(n-k)}\leq \frac{\mathcal{H}_{n-1}}{\Phi(\beta)},
\end{align*}
since $\Phi$ is a non-decreasing function. Finally, since both $\Phi(\beta)$ and $\Phi_{\max}$ are $\Theta(1)$, and using $\ln(n+1) \leq\mathcal{H}_n\leq \ln(n)+1$, we get the second assertion.
\end{proof}

Thus there is a sharp change in the spreading-time between the $SI$ and $SI_{bs}$ epidemic processes in cliques. This is clearly visible in Figure~$1$, where we simulate and plot the spreading-time under the $SI$ and $SI_{bs}$ dynamics on a clique, as a function of network size $n$. Further, in Theorem \ref{appthm:si-lower} in the Appendix, we derive concentration bounds showing that the above result holds with high probability.

\begin{figure}[h]
\label{fig:si-clique}
\begin{center}
\includegraphics[width=\columnwidth]{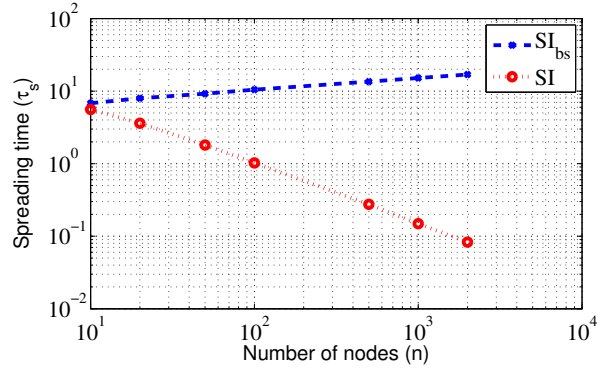}
\caption{Simulated epidemic spreading-time in a clique under $SI$ and $SI_{bs}$, averaged over $500$ runs for $\beta = 0.1$ and $\Phi_{max} = 0.5$. Note that $\EE[\tau_s]$ switches from $o(1)$ under the $SI$ dynamics, to $\Omega(\log n)$ under $SI_{bs}$ (decreasing to increasing on $\log-\log$ plot).}
\end{center}
\end{figure}

On the other hand, the same behavior is not seen in a star network, as we show via our next result:
\begin{proposition}
\label{prop:si-star}
In a star network with a single hub and $n-1$ leaves, under both the $SI$ and $SI_{bs}$ epidemics, we have:
\begin{align*}
\EE[\tau_s]=\Theta(\log n)
\end{align*} 
\end{proposition}

\begin{proof}
Consider first the $SI$ process on a star network. We denote the hub as node $1$ and define: 
\begin{align*}
\tau_1=\inf_{t>0}\{\mbox{Node 1 is infected}\},
\end{align*}
i.e., the first time that the hub gets infected. Since the epidemic starts at a single (arbitrary) node, we can split the spreading-time as: $\tau_s=\tau_{1}+(\tau_s-\tau_1)$. Next, we have that $\EE[\tau_1]=O(1)$ and $\tau_s-\tau_1$ corresponds to the maximum of $m\in\{n-2,n-1\}$ i.i.d random variables, each drawn from a $Exponential(\beta)$ distribution; thus $\EE[\tau_s-\tau_1]=\frac{\mathcal{H}_m}{\beta}\sim\Theta(\log n)$. Finally, note that no susceptible node ever has more than $1$ infected neighbor -- hence the $SI_{bs}$ model is identical to the $SI$ model with $\beta'=\Phi(\beta)=\Theta(1)$. Therefore, the result also holds for $SI_{bs}$ epidemics.
\end{proof}

\subsubsection*{$SIS-SIS_{bs}$ Dynamics}

Next, we consider the $SIS/SIS_{bs}$ epidemic. Here, it is known \cite{massganesh05} that the expected extinction-time is exponential in cliques -- we will prove that it is exponential even for the $SIS_{bs}$ dynamics.

We first state and prove a lemma that gives a closed-form expression for expected time to absorption in a birth-death Markov chain. The lemma is adapted from \cite[Lemma 8]{Banerjee2014a} -- it re-obtains the result for $SIS$ epidemic on cliques, and also extends to our new result for the $SIS_{bs}$ epidemic.

Consider a birth-death CTMC $\{U(t)\}$ on $\{0,1,\ldots,n\}$, with the following transition rates:
\begin{align*}
& i \rightarrow i+1 \quad \text{at rate } q_{i,i+1} \quad \fall 1 \leq i \leq n-1,	\\
& i \rightarrow i-1 \quad \text{at rate }  q_{i,i-1} \quad \fall 1 \leq i \leq n.
\end{align*}
Note that $0$ is the only absorbing state for the above process. Let $T(i,j) := inf\{t: U(t) = j \| U(0) = i\}$. Then we have:
\begin{lemma}[Lemma $8$ from \cite{Banerjee2014a}]
\label{lemma:absorption-time}
\begin{align*}
\EE[T(1,0)] = \frac{1}{q_{1,0}}\left(1 + \sum_{k=2}^n \prod_{i=1}^{k-1} \frac{q_{i,i+1}}{q_{i+1,i}}\right).
\end{align*}
\end{lemma}
\noindent For the sake of completeness, we provide a proof of this result in
the Appendix. We now use this to get lower bounds on the expected extinction-time in cliques.
\begin{proposition}
\label{prop:sis-clique}
For the $SIS/SIS_{bs}$ epidemic on clique $K_n$ starting at an arbitrary node, under both the $SIS$ and $SIS_{bs}$ dynamics, $\log (\EE[\tau_e]) = \Omega(n).$
\end{proposition}

\begin{proof}
We prove the statement for the $SIS_{bs}$ dynamics; the claim for the $SIS$ dynamics follows similarly (and is also known from previous work \cite{massganesh05}). As in Proposition~\ref{prop:si-clique}, it is sufficient to consider the total number of infected nodes $N(t)$ at time $t$. $\{N(t)\}$ has transition rates as follows:
\begin{align*}
& z \rightarrow z+1 \quad \text{at rate } \Phi(\beta z)(n-z) \quad \fall 1 \leq z \leq n-1,	\\
& z \rightarrow z-1 \quad \text{at rate } z \quad \fall 1 \leq z \leq n.
\end{align*}
Let $k_0 = \inf\{k \in \mathbb{N} : \Phi(\beta k) \geq \Phi_{max}/2\}.$ Since $\Phi$ is concave and $\Phi(0) = 0$, we have $\Phi(\beta k) \geq \frac{k}{k_0} \Phi(\beta k_0) \geq \frac{k}{k_0} \frac{\Phi_{max}}{2} \quad \forall k \leq k_0.$ Note that $k_0=\Theta(1)$. Now, from Lemma~\ref{lemma:absorption-time}, we have:
\begin{align*}
\EE[\tau_e] =1 &+ \sum_{k=2}^{n} \prod_{i=1}^{k-1} \frac{\Phi(\beta i)(n-i)}{i+1}	\\
\geq 1 &+ \sum_{k=2}^{k_0} \prod_{i=1}^{k-1}\frac{\Phi_{max}}{2k_0}\frac{i(n-i)}{i+1}\\
& + \sum_{k=k_0+1}^{n} \prod_{i=1}^{k_0-1}\frac{\Phi_{max}}{2k_0} \frac{i(n-i)}{i+1}\prod_{j=k_0}^{k-1} \frac{\Phi_{max}}{2} \frac{(n-j)}{j+1}	\\
= 1 &+\sum_{k=2}^{k_0}\frac{(k-1)!}{n}\binom{n}{k}\left(\frac{\Phi_{max}}{2k_0}\right)^{k-1}\\
& +\frac{k_0!}{nk_0^{k_0}}\sum_{k=k_0+1}^{n}\binom{n}{k}\left(\frac{\Phi_{max}}{2}\right)^{k-1}\\
\geq & \frac{2k_0!}{n\Phi_{max}k_0^{k_0}}\sum_{k=1}^{n} \binom{n}{k}  \left(\frac{\Phi_{max}}{2}\right)^{k}	\\
= &\frac{2k_0!}{n\Phi_{max}k_0^{k_0}} \left(\left(1+ \frac{\Phi_{max}}{2}\right)^n - 1\right).
\end{align*}
Since $\Phi_{max},k_0$ are $\Theta(1)$, $\log\left(\EE[\tau_e]\right) = \Omega(n).$
\end{proof}

The more interesting phenomenon in the $SIS$ model occurs in the star network. As in the clique,it is known \cite{massganesh05,bbcs05} that the expected extinction-time in star networks is exponential. However, under the $SIS_{bs}$ dynamics, we show that it is sub-exponential. More precisely, we show that the epidemic survives in a star network for a time that is \emph{at most polynomial in $n$}.

\begin{proposition}
\label{prop:sis-star}
For an $SIS_{bs}$ epidemic originating at any node in a star graph, $\EE[\tau_e] = O(n^{\Phi_{max}+1}).$
\end{proposition}

This change from exponential to polynomial infection time is clearly seen in simulation results shown in Figure $2$. The proof of this proposition is somewhat technical, so we provide only a proof sketch here. The high-level proof structure is similar to existing analysis of the $SIS$ epidemic on a star \cite{massganesh05,bbcs05} -- the novelty lies in showing that bounded susceptibility causes a sharp transition in the epidemic lifetime.

\begin{proof}[Outline]
We consider cycles of the epidemic evolution -- each cyclical epoch starts off with the hub infected, then recovering to become susceptible, and finally getting reinfected, thereby starting the next cycle. The time from the start of an epoch to when the hub first becomes susceptible is an i.i.d $Exponential(1)$ variable -- at this time, some subset of leaf nodes are infected. Subsequently, either the hub is re-infected (i.e., a new epoch starts), or all these leaf nodes recover before the hub is re-infected -(i.e., the epidemic becomes extinct). \emph{The crux of the proof lies in estimating the probability that, under the $SIS_{bs}$ model, the leaf nodes recover before the hub is reinfected}. Given this, we then show that the total number of such epochs is stochastically dominated by a geometric random variable, whose mean is polynomial in $n$ -- this gives the above result.
\end{proof}

In Section \ref{ssec:results}, we state a much more general version (Theorem \ref{thm:SISbs-starlike}) of this result, extending it to star-like networks. The complete proof of both results are given in Section \ref{sec:proofs}.

\begin{figure}[t]
\label{fig:sis-star}
\begin{center}
\includegraphics[width=\columnwidth]{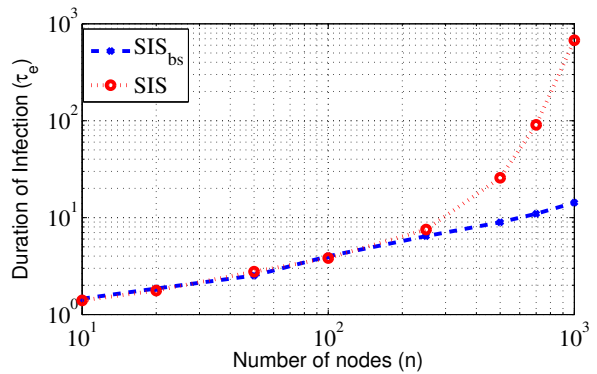}
\caption{Simulated epidemic extinction-time in a star under $SIS$ and $SIS_{bs}$, averaged over $500$ simulation runs for $\beta = 0.1$ and $\Phi_{max} = 0.5$. Note that$\EE[\tau_e]$ switches from exponential under $SIS$ to polynomial under $SIS_{bs}$ (polynomial to linear on $\log-\log$ plot).}
\end{center}
\end{figure}

\subsection{Main Results}
\label{ssec:results}

We now state our main results, which generalize results from the previous section to clique-like and star-like graphs.

\subsubsection{$SI-SI_{bs}$ Dynamics}
\label{ssec:SIresults}

Recall that, for a graph $G$, the Cheeger constant is given by $\eta(G)= \inf_{S \subset V, |S|\leq n/2}\frac{|\delta(S)|}{|S|}$. We now have the following theorem, which allows us to bound the spreading-time of the SI epidemic in clique-like graphs:
\begin{theorem}
\label{thm:SI-upperbnd}
Given underlying graph $G$ with Cheeger constant $\eta(G)$, the spreading-time for an $SI$ epidemic with virulence $\beta$ originating at a single arbitrary node, obeys:
\begin{align*}
\EE[\tau_s]\leq \frac{(2+o(1))\log (n)}{\beta \eta(G)}.
\end{align*}
\end{theorem}

The proof of this theorem, presented in Section \ref{sec:proofs}, is similar to that of Proposition \ref{prop:si-clique}. Note that for the clique, we have $\eta(G) = n/2$. Plugging this into the above bound immediately gives back Proposition \ref{prop:si-clique}. However the theorem captures in finer detail the impact of expansion (as captured by $\eta(G)$) on spreading-time. As a corollary to the theorem, we have the following result for clique-like networks:
\begin{corollary}
\label{cor:SI-cliquelike}
In a network that is clique-like as per Defintion~\ref{def:clique-like}$(A)$ (i.e., $\eta(G)=\Omega(\log n)$), the expected spreading time of an $SI$ epidemic is $O(1)$.
\end{corollary}
We note also that although Theorem~\ref{thm:SI-upperbnd} bounds the \emph{expected} spreading-time, we can extend it to get a concentration bound for the spreading-time; refer to Theorem \ref{appthm:si-upper} in the Appendix for details.
To see how the result can be used to characterize the spreading-time in a variety of clique-like graphs, we present the following corollary  for the case of dense random graphs, i.e., Erd\"{o}s-R\'{e}nyi graphs in the connected regime:
\begin{corollary}
\label{corr:si-gnp}
Let $G$ be an Erd\"{o}s-R\'{e}nyi graph $G(n,p)$, with $p > \frac{32 \log n}{n}$. Then with probability at least $1-\frac{1}{n^2}$, we have $\eta(G) \geq \frac{np}{4}$. Therefore, for an $SI$ epidemic,
$$\EE[\tau_s] \leq \frac{(8 + o(1)) \log n}{\beta np}$$
 with probability at least $1-\frac{1}{n^2}$.
\end{corollary}
Note that by the Borel-Cantelli Lemma, the above bound on the expected spreading-time holds almost surely as $n \rightarrow \infty$. The proof is provided in Section \ref{sec:proofs}.

Putting the pieces together, we can now compare  the $SI$ and $SI_{bs}$ processes in dense random graphs:
\begin{theorem}
\label{thm:si-clique}
Let $G$ be an Erd\"{o}s-R\'{e}nyi graph $G(n,p)$, with $p > \frac{32 \log n}{n}$. For the $SI/SI_{bs}$ epidemic on $G$, starting at an arbitrary node, we have:
\begin{enumerate}[nolistsep,noitemsep]
\item Under the SI dynamics, $\EE[\tau_s] = O\left(\frac{\log n}{np}\right)$ with probability at least $1-\frac{1}{n^2}$.
\item Under the $SI_{bs}$ dynamics $\EE[\tau_s] = \Omega(\log n)$.
\end{enumerate}
\end{theorem}

\begin{proof}
The first claim follows from Corollary \ref{corr:si-gnp} and Theorem \ref{thm:SI-upperbnd}. For the second, note that given any graph $G$, and $G'$ obtained by adding edges to $G$, $\tau_s(G')\leq_{st}\tau_s(G)$ -- this can be shown via a standard coupling argument. Thus, for any graph $G$, we can always lower bound the spreading-time under the $SI_{bs}$ dynamics by that on the clique $K_n$. The second claim thus follows from Proposition \ref{prop:si-clique}.
\end{proof}


Next we turn to the case of star-like graphs. To characterize $SI$ epidemics on such networks, we first state a lemma that gives a lower bound on $\EE[\tau_s]$ in terms of average degree:
\begin{lemma}
\label{lem:SI-starlike}
Consider a graph $G$ with $|V|=n$, and an $SI$ epidemic on $G$, starting at some node $v_1$. Let $A\subseteq V$ be a subset of nodes containing $v_1$, and let $\widehat{d}_{A^c}$ be the average degree of the nodes in $V\setminus A$, i.e., $\widehat{d}_{A^c} = \frac{\sum_{v\notin A}d(v)}{n-|A|}$. Then:
\begin{align*}
\EE[\tau_s]\geq \frac{\log (n-|A|)}{\beta\widehat{d}_{A^c}}
\end{align*}
\end{lemma}
For example, for the star graph with $v_1$ chosen to be the hub, we have $\widehat{d}_{\{v_1\}^c}=1$, and so we get back the lower bound in Proposition \ref{prop:si-star}. We now use the above lemma to show that imposing bounded susceptibility does not significantly change the spreading-time in star-like networks:

\begin{theorem}
\label{thm:SI-starlike}
Consider a network $G$ of diameter $O(\log n)$ that is star-like by Definition~\ref{def:star-like}$(A)$.  In other words, there exists a partition of the nodes into a core $V_c$ and periphery $V_p$, parametrized by:
\begin{itemize}[nolistsep,noitemsep]
\item \emph{Core size:} $|V_c|=m$ is at most a constant fraction of nodes (i.e. $m\leq cn$ for some $c<1$). In other words, $|V_p|= \Omega(n)$.
\item \emph{Sparse periphery:} $\widehat{d_p}=O(1)$, i.e., the average degree of nodes in the periphery is constant (or lesser).
\end{itemize}
 Then both under the $SI$ and $SI_{bs}$ dynamics, $\EE[\tau_s]=\Theta\left(\log n\right)$.
\end{theorem}

Note that the sparsity constraint is on the average degree of nodes in $V_p$, not in the subgraph induced by $V_p$, i.e., it also accounts for the edges in the cut $\delta(V_p)$. 

\begin{proof}
First consider the $SI$ dynamics, starting at some arbitrary node $v_1$.
We want to derive a lower bound on the spreading-time using Lemma \ref{lem:SI-starlike} -- to do so, we choose $A=V_c\cup\{v_1\}$, and hence $A^c=V_p\setminus\{v_1\}$. Now using Lemma \ref{lem:SI-starlike} and the fact that $|V_p|\sim poly(n)$, and that both $\beta$ and $\widehat{d}_p$ are $\Theta(1)$ we have that $\EE[\tau_s]=\Omega(\log n)$. 

To obtain an upper bound on the spreading-time, we under-dominate the spread by considering epidemic spread on the shortest-path spanning tree of the network rooted at the initial infected node. The resulting depth is at most the diameter of the tree. Thus, the spreading-time is stochastically upper bounded by the maximum of $n$ i.i.d random variables distributed as sum of $diam(G)$ number of exponential random variables with rate $\beta$. Finally, using standard concentration results for the sum of exponentials (refer \cite{gopban11}), we have that $\EE[\tau_s]=O(\log n)$.

For $SI_{bs}$ epidemic, we can first over-dominate the epidemic spread by converting the network into a clique -- we have from Proposition \ref{prop:si-clique} that the time taken to infect all nodes is $\Omega(\log n)$. On the other hand, we can again under-dominate the spread by considering an $SI_{bs}$ process on the minimum spanning tree which is equivalent to an $SI$ process of rate $\Phi(\beta)$. Consequently, $\EE[\tau_s]=\Theta(\log n)$. 
\end{proof}

\subsubsection{$SIS-SIS_{bs}$ Dynamics}
\label{ssec:SISresults}

Next we look at the $SIS-SIS_{bs}$ dynamics on clique-like and star-like graphs. Recall that in this case, we are interested in the extinction-time $\tau_e$. In particular, we are interested in characterizing networks in which the epidemic is \emph{persistent}, i.e. the extinction-time is exponential in $n$ (i.e., $\EE[\tau_e]\sim\Omega(e^{poly(n)})$).

Existing works \cite{massganesh05,bbcs05,Banerjee2014a} demonstrate that the virulence thresholds for short-lived ($\EE[\tau_e]\sim O(\log n)$) $SIS$ epidemics depend on the spectral radius, $\rho(G)$,  while thresholds for persistent epidemics are related to the generalized isoperimetric constants. Though the results for the SIS epidemic depend on edge-expansion isoperimetric constant $\eta^e_m(G)$, the vertex-expansion isoperimetric constant $\eta^v_m(G)$ turns out to be the appropriate notion for the $SIS_{bs}$ epidemic, as follows:

\begin{theorem}
\label{thm:SIS-clique}
Consider an $SIS_{bs}$ epidemic on graph $G$, with virulence $\beta$ and infection-profile $\Phi(\cdot)$. Suppose for $m=n^{\alpha}$ for some $0<\alpha <1$, we have that:
\begin{align*}
\Phi(\beta) \eta_{m}^v(G)> 1,
\end{align*}
then $\EE[\tau_e] = e^{\Omega(n^\alpha)}$, i.e., the $SIS_{bs}$ epidemic is persistent .
\end{theorem}

\begin{proof}[Outline]
The proof is similar to that of Proposition \ref{prop:sis-clique}. In particular, we construct an under-dominating CTMC for the number of infected nodes under $SIS_{bs}$. Next, we use Lemma \ref{lemma:absorption-time} to get an expression for the extinction-time, and show that the bound is indeed exponential in $n$ if $\Phi(\beta) \eta_{m}^v(G)> 1$. For complete proof, see Section \ref{sec:proofs}.
\end{proof}
An immediate consequence of this result is the following corollary for clique-like networks.
\begin{corollary}
\label{cor:SIS-cliquelike}
An $SIS_{bs}$ epidemic is persistent in any network belonging to the class of clique-like networks as per Definition~\ref{def:clique-like}$(B)$. 
\end{corollary}
Returning to the $SIS$ epidemic, the strongest known condition for a persistent $SIS$ epidemic is $\beta\eta^e_m(G) >1$ \cite{massganesh05}. Now since $\eta^e_m(G)>\eta^v_m(G)$, this implies that \emph{$\beta\eta^v_m(G)>1$ is also sufficient for a persistent $SIS$ epidemic}, i.e., without bounded susceptibility. Theorem \ref{thm:SIS-clique} shows that the condition for a persistent $SIS_{bs}$ epidemic is closely related. Furthermore, as we show subsequently, the vertex-expansion condition is strong enough to recover all the results from \cite{massganesh05}, in particular for the $G(n,p)$ and power-law graphs.

We now use Theorem~\ref{thm:SIS-clique} to demonstrate that both $SIS$ and $SIS_{bs}$  epidemics have exponential extinction-times in dense random graphs:
\begin{theorem}
\label{thm:sis-gnp}
For both the $SIS$ and $SIS_{bs}$ epidemics on a $G(n,p)$ graph with $p > \frac{16\log n}{n}$. For a large enough $n$, with probability at least $1-\frac{1}{n^2}$ we have,
$$\log\left(\EE[\tau_e]\right) = \Omega(n).$$
\end{theorem}

\begin{proof}[Outline]
We show under the conditions of the theorem that $\Phi(\beta) \eta_{m}^v(G)> 1$ with probability $\geq 1-\frac{1}{n^2}$ if $m=n^{\alpha}$ for any constant $\alpha \in (0,1)$. The result then follows using Theorem~\ref{thm:SIS-clique}. Complete proof is given in Section~ \ref{sec:proofs}.
\end{proof}

Note also that any graph that contains a large clique-like subgraph is prone to an exponentially long lasting infection, even under bounded susceptibility, if the initial infected node belongs to the subgraph. Clearly a clique, or a realization of an Erd\"{o}s-R\'{e}nyi graph in the connected regime experiences long lasting infection -- however such dense random subgraphs appear in other graph families including `power-law' or scale-free networks. We discuss this further in Section \ref{ssec:disc} (in particular, Corollary \ref{cor:power-law}).

%

More generally, these results show that large dense networks experience persistent epidemics even under bounded susceptibility \emph{due to the presence of large clique-like subgraphs}. This is in sharp contrast to existing work (for example, \cite{bbcs05}) which suggests that persistent epidemics in such networks arise due to the presence of high-degree nodes (`hubs').

In fact, it turns out that the presence of such `hubs' do not help -- we took a step towards showing this in Proposition \ref{prop:sis-star}, where we showed that a star graph does not experience persistent $SIS_{bs}$ epidemics. We now show that this is true for a much larger class of star-like graphs. As before, the node-set $V$ is decomposed into core $V_c$ and periphery $V_p$. We characterize the sparsity of the periphery via the \emph{spectral radius} $\rho(G(V_p))$ of the periphery subgraph -- recall that $d_{avg}(G(V_p))\leq\rho(G(V_p))\leq d_{\max}(G(V_p))$. Now we have:

\begin{theorem}
\label{thm:SISbs-starlike}
Consider a network $G$, with the nodes partitioned into a core $V_c$ and periphery $V_p$ characterized by:
\begin{itemize}[nolistsep,noitemsep]
\item \emph{Core size:} $|V_c|=m$.
\item \emph{Periphery sparsity:} Parametrized by $\rho(G(V_p))$.
\end{itemize}
Suppose $G(V_p)$ satisfies $\Phi(\beta)\rho(G(V_p))= 1-\epsilon$, for some $\epsilon > 0$. Then for an $SIS_{bs}$ epidemic on $G$,:
\begin{align*}
\log\left(\EE[\tau_e]\right)= O(m\log n).
\end{align*}
\end{theorem}

\begin{proof}[Outline]
The high-level proof structure is similar to the outline we gave for Proposition \ref{prop:si-star}. We again consider cyclic epochs. An epoch starts with some node in the core infected, transitioning to all core nodes becoming susceptible, and finally ends when some node in the core gets reinfected beginning the next epoch. Next, we bound under the $SIS_{bs}$ model the probability that, in any epoch, the epidemic dies out in all nodes of the periphery before any node in the core is reinfected. Finally, via stochastic domination by a geometric random variable, we get the above result. The complete proof is given in Section \ref{sec:proofs}.
\end{proof}
Though we state the above theorem for a general core size, it is most significant in settings the core is small. For example, if the core is of constant size, then \emph{$\EE[\tau_e]$ scales only polynomially with $n$} -- note that this includes the star network (where $\rho(G(V_p))=0$, as the periphery is fully disconnected if we remove the hub) that we characterized in Proposition~\ref{prop:sis-star}. More generally, the following result for star-like networks holds as a result of Theorem~\ref{thm:SISbs-starlike}.
\begin{corollary}
\label{cor:SISbs-starlike}
The extinction time of an $SIS_{bs}$ epidemic is sub-exponential in any network that is star-like as per Definition~\ref{def:star-like}$(B)$.
\end{corollary}
We note that we can also derive concentration bounds, showing that the bound on the order of the extinction-time holds with high probability (see Theorem \ref{appthm:sisstar} in the Appendix).


\subsection{Discussion}
\label{ssec:disc}

We now show how our results apply to the settings we discuss in Section \ref{ssec:intro-appl}.

\paragraph*{Social and Economic Networks}

Epidemic models have been widely used to model the viral spread of content in social networks. In particular, the $SIS$ model can be used to study the spread of `memes', i.e., related pieces of content that re-appear cyclically. Furthermore, the exponential extinction-time of an epidemic is in a sense connected to long-lasting nature of popular memes.

However, the existence of popular memes is often regarded to be as a result of highly influential nodes, i.e., those with high degree. Such high-degree nodes are a distinctive feature of power-law graphs, which are often used to model such networks. Our results however suggest that persistent epidemics arise due to large random subgraphs rather than high-degree nodes. We now formalize this notion by applying Theorem~\ref{thm:sis-gnp} to power-law graphs.

We consider the case of power law graphs with exponent $2< \gamma < 3$. For random power law graphs generated under the Chung-Lu model \cite{chung-etal03eigenvalues-power-law} with exponent $2< \gamma < 3$ , it is known that they contain a $n^{\epsilon_1}$ connected random graph \cite{massganesh05}, and further, that they contain an $n^{\epsilon_2}$ clique  \cite{janson10large-clique-power-law} for some $\epsilon_1, \epsilon_2 > 0$. The $SIS$ epidemic is known to have exponential lifetime in such graphs \cite{massganesh05}; using Theorem \ref{thm:sis-gnp}, we recover this behavior under bounded susceptibility:
\begin{corollary}
\label{cor:power-law}
For networks generated via the Chung-Lu model \cite{chung-etal03eigenvalues-power-law}, if the power-law exponent $\gamma \in (2,3)$ and $\Phi_{max} > 1 + v$ for some $v > 0$, the $SIS_{bs}$ epidemic has exponentially long lifetime.
\end{corollary}
Thus scale-free networks do experience persistent epidemics \emph{even under bounded susceptibility}, due to the presence of large clique-like subgraphs. \emph{This questions the belief that highly-connected nodes are most responsible for viral content.}

\paragraph*{Epidemiology}
Epidemic models help understand the dynamics of infectious diseases amongst human/animal populations, and the spread of computer viruses and worms. These models also help in the design of defense mechanisms against such epidemics -- for example, policies for spreading \emph{antidotes}, which increase the recovery rate in the $SIS$ epidemic \cite{BorgsAntidote}. In this regard, our work helps in developing a theory for \emph{partial quarantine policies}, wherein adopting preventative measures helps reduce (and bound) the susceptibility of a node. We make a distinction here between partial and \emph{perfect} quarantine policies, wherein a node can be completely disconnected from the network. Achieving a perfect quarantine is often unfeasible and there is always some residual susceptibility.

In this regard, our results paint a somewhat bleak picture -- \emph{enforcing quarantine policies in dense networks has an insignificant effect on the epidemic lifetime}. More specifically, as long as there is some $\Theta(1)$ residual susceptibility, the epidemic lifetime remains exponential in clique-like networks. Note also that a similar effect is seen in \cite{BorgsAntidote} (under linear susceptibility), where an effective policy needs to give an amount of antidote to a node which is proportional to its degree -- this is clearly more difficult in dense graphs.

The story is somewhat more positive in star-like networks, where Theorem~\ref{thm:SISbs-starlike} shows that in a network with a small core and a sparse outer network, partial quarantining affects a large change in the extinction-time, reducing it to $poly(n)$. Star-like networks are relevant for studying the spread of an infectious disease through human transportation networks, which comprise a core of long distance air-links, overlaid on a slower road network. In case of an epidemic, a natural policy is to focus on safeguarding the core cities (via quarantining). For this, we have the following corollary:
\begin{corollary}
\label{cor:overlay}
Consider a network which can be decomposed into:
\begin{itemize}[nolistsep,noitemsep]
\item A dense core of size $O(1)$, which can be fully connected.
\item A sparse periphery, with $\rho(G(V_p))=O(1)$ (no restriction on edges between periphery and core).
\end{itemize} 
Now suppose \emph{only core nodes have bounded susceptibility} (via quarantining), and the unprotected periphery satisfies $\beta\rho(G(V_p))\leq 1-\epsilon$ for some $\epsilon > 0$. Then the epidemic is not persistent, i.e., $\EE[\tau_e]=O(poly(n))$.
\end{corollary}
Note that given $\Omega(n)$ edges between the hub and the periphery, the $SIS$ epidemic has an exponential lifetime. Thus, safeguarding hub cities does help reduce the epidemic lifetime dramatically. However, \emph{distributing antidotes may be a much more effective policy than quarantining as it can effect much faster (upto $O(\log n)$) extinction-times.}

\paragraph*{Distributed Computing Applications}
Another important application of SI epidemics is as a primitive for algorithms like \emph{flooding} \cite{Eugster04}, i.e., forwarding a message to all network nodes. In order to minimize the flooding time, one approach would be to make the underlying network clique-like -- however our results indicate that under bounded susceptibility, the flooding time is $\Omega(\log n)$, which is identical to that in a star-like network (for example, a hierarchical network, with a constant-sized core and all other nodes connected to the core). On the other hand, our results also show that star-like networks are much easier to protect from viruses which spread like an SIS process. Considering the tradeoff between performance and robustness to infection, \emph{star-like architectures appear more appropriate than clique-like architectures for distributed computing applications}.

In case of networks using SIS epidemics to ensure persistent data storage  \cite{Chakrabarti07}, the story is strikingly different from the above. Here, the presence of bounded susceptibility results in a sharp difference between epidemic lifetime in star-like and clique-like networks. However, our results also indicate that it is not necessary to have very high edge densities, or very well-designed networks in order to achieve long lifetimes -- on the contrary, \emph{as long as there exists a large-enough connected random subgraph, the epidemic lifetime is exponential even under bounded susceptibility}.

\section{Proofs}
\label{sec:proofs}

We now give detailed proofs for all our results. We use the standard
stochastic dominance notation $\{\leq_{st},\geq_{st}\}$ for relations
between random variables. Specifically, $(X \leq_{st} Y) \implies
\forall z, \mathbb{P}[X > z] \leq \mathbb{P}[Y > z].$

%


\paragraph{Proofs from Section \ref{ssec:SIresults}}


\begin{proof}[of Theorem \ref{thm:SI-upperbnd}]
Let $N(t)$ be the number of infected nodes at time $t$ and $T_k$ be the first time at least $k$ nodes are infected, i.e., $T_k = \inf \{t : N(t) = k\}$. Also, $\fall k\in [n]$, let $\mathcal{V}_k$ denote the set of all subsets of $V$ of size $k$, and let $S_k\in \mathcal{V}_k$ be the set of infected nodes at $T_k$. 

Since $T_1 = 0$ and $T_n = \tau_s$, we can telescope and write $\tau_s = \sum_{k=1}^{n-1}(T_{k+1} - T_k)$. Recall that in case of the clique in Proposition \ref{prop:si-clique}, we had that $\EE[T_{k+1} - T_k]=\frac{1}{\beta k(n-k)}.$ This is no longer true, but we can upper bound the time in terms of the isoperimetric constant $\eta(G)$, as follows.
\begin{align*}
\EE[T_{k+1}-T_k] &= \sum_{s_k\in\mathcal{V}_k}\EE[T_{k+1}-T_k|S_k=s_k]\PP[S_k=s_k]\\
& = \EE\left[\frac{1}{\beta|\delta(S_k)|}\right].
\end{align*}
By the definition of $\eta(G)$, 
\[
|\delta(S_k)| \geq 
\begin{cases}
 \eta(G)k & \text{if } k < n/2	\\
 \eta(G)(n-k) & \text{if } k \geq n/2.	
\end{cases}
\]
Therefore, we have
\begin{align*}
\EE[\tau_s] 
&\leq \sum_{k=1}^{\lfloor n/2 \rfloor} \frac{1}{\beta\eta(G) k} + \sum_{\lfloor n/2 \rfloor + 1}^{n-1} \frac{1}{\beta\eta(G) (n-k)} \\
&\leq 2 \sum_{k=1}^{\lfloor n/2 \rfloor} \frac{1}{\beta\eta(G) k} \leq \frac{(2+o(1))\log (n/2)}{\beta\eta(G)}.
\end{align*}
This completes the proof.
\end{proof}

Next we show that a $G(n,p)$ in the connected regime has a large isoperimetric constant with high probability.

\begin{proof}[of Corollary \ref{corr:si-gnp}]
Let graph $G$ be a realization of a $G(n,p)$ random graph with $p > \frac{32 \log n}{n}$. We claim that
\begin{align*}
\label{eq:tau-gnp}
\PP\left[\eta(G) < np/4\right] \leq \frac{1}{n^2}.
\end{align*}
In other words, for all but a vanishing fraction of graphs (of measure $\leq 1/n^2$), we have the required bound on $\eta(G)$. To show this, we write
\begin{align*}
\PP\left[\eta(G) < np/4\right] & = \PP\left[\bigcup_{k \leq \lfloor n/2 \rfloor} \bigcup_{s_k\in\mathcal{V}_k} |\delta(s_k)| < nkp/4 \right]	\\
& \leq \PP\left[\bigcup_{k \leq \lfloor n/2 \rfloor} \bigcup_{s_k\in\mathcal{V}_k} |\delta(s_k)| < k(n-k)p/2 \right]	\\
& \leq \sum_{k \leq \lfloor n/2 \rfloor} \sum_{s_k\in\mathcal{V}_k} \PP\left[|\delta(s_k)| < k(n-k)p/2 \right].
\end{align*}
We can now use the following Chernoff bound -- for any r.v. $B$ taking values in $[0,1]$, we have
$$\PP\left[B < (1-\delta)\EE[B]\right] \leq \exp\left(\frac{-\EE[B] \delta^2}{2}\right).$$
Since $|\delta(s_k)|\sim Binom(k(n-k), p)$, 
\begin{align*}
\PP\left[\eta(G) < np/4\right] & \leq \sum_{k=1}^{\lfloor n/2 \rfloor} \binom{n}{k} \exp\left(\frac{-1}{8}k(n-k)p\right)\\
& \leq n \binom{n}{1} \exp\left(\frac{-1}{8}(n-1)p\right).
\end{align*}
Using the condition $p > \frac{32 \log n}{n}$, we get our result.
\end{proof}

\noindent Next, we turn to Lemma \ref{lem:SI-starlike}, wherein we derive a lower bound on the spreading-time in terms of the average degree of a graph. First, we need the following lemma:
\begin{lemma}
\label{lem:max-exp-stoc-domin}
Let $\{X_i\}_{i=1}^m$ be independent exponential random variables with rate $a_i$ (i.e., $Exponential(a_i)$) such that $\sum_{i=1}^m a_i = A,$ and let $\{Y_i\}_{i=1}^m$ be i.i.d exponential random variables with rate $A/m$. If $X_{\max} = \max\{X_1, X_2, \dots , X_m\}$ and $Y_{\max} = \max\{Y_1, Y_2, \dots , Y_m\},$ then $X_{\max} \geq_{st} Y_{\max}$.
\end{lemma}

\begin{proof}
For stochastic dominance, we want to show that $\PP\left[X_{\max} \leq x \right] \leq \PP\left[Y_{\max} \leq x \right] \fall x > 0$. For any $x > 0$, define $Q_x(\mathbf{a}) = \prod_{i=1}^m (1-\exp(-a_ix)).$ We now show that $\mathbf{a^*} = (A/m, A/m, \dots, A/m)$ maximizes $Q_x$ over the set $\mathcal{A} = \{\mathbf{a}: \sum_{i=1}^m a_i = A\}$. For $m=2$, one can check that $\mathbf{a^*} = (A/2, A/2)$ is the unique maximum. Thus, $Q_x(\mathbf{a^*}) > Q_x(\mathbf{a})$ for any $\mathbf{a} \neq \mathbf{a^*}.$ If $m > 2,$ suppose that $\tilde{\mathbf{a}} \in \mathcal{A}$ and $\tilde{\mathbf{a}} \neq \mathbf{a^*},$ maximizes $Q_x$ over $\mathcal{A}.$ Then $\exists i,j$ such that $\tilde{a}_i \neq \tilde{a}_j.$ Then, for $\hat{\mathbf{a}}$ equal to $\tilde{\mathbf{a}}$ except with the $i, j$ elements replaced with $(\tilde{a}_i +\tilde{a}_j)/2,$ we have $Q_x(\hat{\mathbf{a}}) > Q_x(\tilde{\mathbf{a}}),$ which is a contradiction. Therefore, $\mathbf{a^*} = (A/m, A/m, \dots, A/m)$ maximizes $Q_x$ over the set $\mathcal{A},$ and $\PP\left[Y_{\max} \leq x \right] = Q_x(\mathbf{a^*}) \geq Q_x(\mathbf{a}) = \PP\left[X_{\max} \leq x \right]$ for any $x > 0$.
\end{proof}

\noindent Using this, we now prove Lemma \ref{lem:SI-starlike}.

\begin{proof}[of Lemma~\ref{lem:SI-starlike}]
Let the nodes be numbered from $1$ to $n$ such that $A = \{1,2,\dots, |A|\}$ and let $\mathbf{d} = (d_1, d_2, \dots, d_n)$ be the degree sequence of $G$. At any time, the rate at which a susceptible node $i$ is infected is upper bounded by $\beta d_i$. Therefore, $\tau_s \geq_{st} \max\{X_{|A|+1}, \dots , X_n\}$ where $\{X_i\}$ are independent exponential random variables with mean $1/\beta d_i.$ By Lemma~\ref{lem:max-exp-stoc-domin}, $\tau_s \geq_{st} Y,$ where $Y$ is the maximum of $n-|A|$ i.i.d. exponential random variables with mean $\frac{n-|A|}{\sum_{i=|A|+1}^n \beta d_i} = \frac{1}{\beta\widehat{d}_{A^c}}$.
\begin{align*}
\EE\left[\tau_s\right] \geq \EE\left[Y\right] = \frac{\mathcal{H}_{n-|A|}}{\beta\widehat{d}_{A^c}} > \frac{\log (n-|A|)}{\beta\widehat{d}_{A^c}}.
\end{align*}
This completes the proof.
\end{proof}

\paragraph{Proofs from Section \ref{ssec:SISresults}}

\begin{proof}[of Theorem~\ref{thm:SIS-clique}]
We analyze a CTMC $Z(t)$ that is stochastically dominated by the total number of infected nodes $N(t)$ in the $SIS_{bs}$ epidemic. This being true, we have that the expected time to absorption for $Z(t)$ is less than that of $N(t)$. We note that this proof structure is similar to that of Theorem $4.1$ from \cite{massganesh05}; however it is much simplified via the ergodic embedding technique of \cite{Banerjee2014a}.

Note that for any set of nodes $A\subseteq V$ of size $|A|\leq m$ (for some $m = n^\alpha$), we have that the neighborhood of $A$ has size at least $\eta_m^v(G)|A|$. Now suppose the set of infected nodes $I$ satisfies $|I|\leq m.$ Under bounded susceptibility, each node in $\Gamma(I)$ is infected at a rate which is $\geq \Phi(\beta)$. Using this observation, we define $Z(t)$ to have transition rates given by
\begin{align*}
& z \rightarrow z+1 \quad \text{at rate } \eta_m^v(G)\Phi(\beta) z \quad \fall 1 \leq z \leq m-1,	\\
& z \rightarrow z-1 \quad \text{at rate } z \quad \fall 1 \leq z \leq m,
\end{align*}
and all other transitions equal to $0$. Note that $Z(t)$ is a finite state space CTMC on state space $\{0,1,\ldots,m\}$ -- standard coupling arguments (for example, see \cite{DraiefMass}) show that it is dominated by $N(t)$. Now using Lemma~\ref{lemma:absorption-time}, we have
\begin{align*}
\mathbb{E}[\tau_e] & \geq \mathbb{E}[\inf\{t:Z(t) = 0\}] = \sum_{k=0}^{m-1} (\eta^v_m(G)\Phi(\beta))^k.
\end{align*}
So, if $\Phi(\beta)\eta_m^v(G) > 1$, we have $\log (\mathbb{E}[\tau_e]) = \Omega(n^{\alpha}).$
\end{proof}

Now we show that for the Erd\"{o}s-R\'{e}nyi graph in the connected regime, the condition for persistent $SIS_{bs}$ epidemics is satisfied with high probability. Note that this proof also immediately gives persistence for the $SIS$ epidemic.

\begin{proof}[of Theorem ~\ref{thm:sis-gnp}]
Let $G$ be a realization of a $G(n,p)$ graph. Fix $\alpha \in (0,1), m=n^{\alpha}$, and choosing some $\gamma>1$, we define $b=\gamma/\Phi(\beta)$. We now prove that $\PP\left[\eta^v_m(G) < b\right] \leq 1/n^2$ for $p=\frac{16\log n}{n}.$ Then by stochastic dominance, the statement is true for any $p > \frac{16\log n}{n}$. First we have
\begin{align*}
\PP\left[\eta^v_m(G) < b\right] \leq \sum_{i=1}^m \sum_{S:|S|=i} \PP\left[|\Gamma(S)| < bi\right].
\end{align*}
Defining $B(n,p,i)\sim Binomial(n-i,1-(1-p)^i)$, we can write
\begin{align*}
\PP\left[\eta^v_m(G) < b\right] \leq \sum_{i=1}^m \binom{n}{i} \PP\left[B(n,p,i) < bi\right].
\end{align*}

Note that $(1-p)^i\leq 1-ip+i^2p^2/2 \fall i$. This implies that
$$\EE[B]\geq (n-i)ip\left(1-ip/2\right).$$
Defining $\delta =\left(1-\frac{bi}{(n-i)ip(1-ip/2)}\right)$ and using Chernoff bound,
\begin{align*}
\PP\left[B(n,p,i) < bi\right]& \leq \PP\left[B(n,p,i)<(1-\delta)\EE[B(n,p,i)]\right]\\
&\leq \exp\left(\frac{-\EE[B(n,p,i)]\delta^2}{2}\right)\\
&\leq \exp\left(\frac{-(n-i)ip(1-ip/2)\delta^2}{2}\right).
\end{align*}
Note that $\frac{(n-i)}{n}(1-ip/2)\delta^2 \rightarrow 1$ as $n \rightarrow \infty.$ So for $n$ large enough, $(n-i)(1-ip/2)\delta^2 > n/2$ which implies
\begin{align*}
\PP\left[\eta^v_m(G) < b\right] & \leq \sum_{i=1}^m \binom{n}{i} \exp\left(-\frac{1}{2}(n-i)ip(1-ip/2)\delta^2\right)\\
& \leq \sum_{i=1}^m \frac{n^i}{i!}  \exp(-i(np/4))	\\
& = \sum_{i=1}^m \frac{1}{i!} \exp(-i(np/4 - \log n)) < \frac{1}{n^2}.
\end{align*}
\end{proof}

Next we turn to Theorem \ref{thm:SISbs-starlike}. The high-level proof structure is similar to existing analysis for the SIS epidemics on stars \cite{massganesh05,bbcs05}. However, to study the effect of bounded susceptibility, we need to carefully characterize various inter-epoch events, which require controlling the size of the core and the sparsity of the periphery.

We need two crucial lemmas for the proof. The first lemma is a crude but general upper bound for $SIS_{bs}$ extinction-time in \emph{all graphs},
\begin{lemma}
\label{lemma:sis-clique-upper-bd}
For any network of size $n$, and for any arbitrary initial set of infected nodes, the expected extinction-time for $SIS_{bs}$ dynamics has an upper bound given by
$$\mathbb{E}[\tau_e] \leq \frac{(1+\Phi_{max})^n}{n \Phi_{max}}.$$
\end{lemma}
\begin{proof}
We use the expression for $\mathbb{E}[\tau_e]$ from Lemma~\ref{lemma:absorption-time}. For any susceptible node with an infected neighbor, we have that the infection rate is bounded by $\Phi_{max}$. Further, at any time when the number of infected nodes $|I|\geq 1$, we can stochastically upper bound the infection by assuming that \emph{all susceptible nodes} have an infected neighbor. Thus
\begin{align*}
\mathbb{E}[\tau_e] & \leq 1 + \sum_{k=2}^{n} \prod_{i=1}^{k-1} \frac{\Phi_{max}(n-i)}{i+1}	\\
& = \frac{1}{n \Phi_{max}} \sum_{k=1}^{n} \binom{n}{k} \Phi_{max}^k	
\leq \frac{(1+\Phi_{max})^n}{n \Phi_{max}}.
\end{align*}
\end{proof}

We use Lemma \ref{lemma:sis-clique-upper-bd} to control extinction-time in dense subgraphs. In case of sparse subgraphs, we can use a result from \cite{massganesh05}, which gives conditions for fast extinction of the $SIS$ epidemic (and hence $SIS_{bs}$ epidemic) in terms of $\rho(G)$.
\begin{lemma} \cite[Theorem 3.1]{massganesh05}
\label{lemma:tau-sub-critical}
Let $\rho(A)$ be the spectral radius of the adjacency matrix $A$. If $\beta < 1/\rho(A)$, then the probability that the epidemic has not died out by time $t$, given the initial condition $X(0)$, admits the following upper bound.
\begin{equation*}
\mathbb{P}\left[X(t) \neq 0\right] \leq \sqrt{n \lVert X(0) \rVert_1} e^{-(1-\beta \rho(A))t},
\end{equation*}
where $\lVert X(0) \rVert_1 = \sum_{i=1}^n X_i(0)$.
\end{lemma}
\noindent With these two ingredients, we can now prove Theorem~\ref{thm:SISbs-starlike}.

\begin{proof} [of Theorem~\ref{thm:SISbs-starlike}]
We define cyclical epochs, starting with some (arbitrary) set of nodes in the core $V_c$ infected, transitioning to all nodes in $V_c$ becoming susceptible, and then finally at least one node in $V_c$ getting reinfected. This technique is similar those used in \cite{massganesh05,bbcs05} -- however instead of a star, we now have to deal with \emph{core and periphery subgraphs} and also \emph{bounded susceptibility.}

Consider epoch $j$ (for any $j\in\setN$) starting at time $T_{0,j}$. We define $T_{1,j} := \inf\{t > T_{0,j}: \sum_{i \in V_c}X_i(t) = 0\}$, i.e., all nodes in $V_c$ become uninfected. From Lemma~\ref{lemma:sis-clique-upper-bd}, we have
\begin{equation}
\label{eq:ET1-bound}
\EE[T_{1,j}-T_{0,j}] \leq \frac{(1+\Phi_{max})^m}{m \Phi_{max}},
\end{equation}
where $m=|V_c|$. Note that, under bounded susceptibility, this \emph{only depends on $m$ and not on core/periphery topologies}.

Next, we define $T_{0,j+1}$ to be the first time after $T_{1,j}$ when either some node in $V_c$ gets re-infected or all the nodes in $V_p$ recover -- in case of the former, the next epoch starts, while in case of the latter, the infection dies out (i.e., $T_{0,j+1}=\tau_e$).

It is clear that the random variable $T_{0,j+1}-T_{1,j}$ is stochastically dominated by $\zeta_1$, the time taken by all the nodes in $V_p$ to completely recover if the periphery was isolated from the core at $T_{1,j}$. Now given that the infection profile $\Phi$ is concave, and $\Phi(0) = 0$, we have $\Phi(\beta x) \leq \Phi(\beta)x \quad \forall x \geq 1$ and so the number of infected nodes under the $SIS_{bs}$ dynamics is stochastically dominated by that in  an $SIS$ infection process with rate $\Phi(\beta)$. Now given that $\Phi(\beta)\rho(G(V_p))=1-\epsilon$ for some $\epsilon>0$, where $\rho(G(V_p))$ is the spectral radius of the periphery subgraph, we have from Lemma~\ref{lemma:tau-sub-critical} that $\mathbb{P}(\zeta_1 > t) \leq \min\left\{ne^{-(1 - \Phi(\beta)\rho(G(V_p)))t}, 1\right\}$. Thus
\begin{align}
\label{eq:ET2-bound}
\EE[T_{0,j+1}-T_{1,j}] & \leq \EE[\zeta_1] = \int_0^{\infty} \mathbb{P}(\zeta_1 > t)dt	\nonumber \\
& \leq \int_0^{\infty} ne^{-(1 - \Phi(\beta)\rho(G(V_p)))t} dt	\nonumber \\
& = \frac{n}{1 - \Phi(\beta)\rho(G(V_p))}.
\end{align}

We now derive a lower bound for the probability $q$ that the infection dies in the periphery $G(V_p)$ before any node in the core gets reinfected. Since each node in the core $V_c$ is infected at a rate at most $\Phi_{max}$, the time after $T_{1,j}$ for the core to get infected can be bounded from below by $\zeta_2,$ an $Exponential(m\Phi_{max})$ random variable, where $m=|V_c|$.
Defining $t_0 = \frac{\log n}{1-\Phi(\beta) \rho(G(V_p))}$, we have
\begin{align*}
\mathbb{P}[\zeta_1 > \zeta_2]& = m\Phi_{max}\int_0^{\infty} \mathbb{P}(\zeta_1 > t) e^{-m\Phi_{max}t}dt	\\
& \leq m\Phi_{max}\int_0^{t_0} e^{-m\Phi_{max}t}dt\ldots\\
& +  m\Phi_{max}\int_{t_0}^{\infty} ne^{-(1 - \Phi(\beta)\rho(G(V_p)) + m\Phi_{max})t}dt	\\
& = 1 - e^{-m\Phi_{max}t_0}\ldots \\
&  + \frac{nm\Phi_{max}e^{-(1 - \Phi(\beta)\rho(G(V_p)) + m\Phi_{max})t_0}}{1 - \Phi(\beta)\rho(G(V_p)) + m\Phi_{max}}	\\
& = 1 - \frac{1 - \Phi(\beta)\rho(G(V_p))}{1 - \Phi(\beta)\rho(G(V_p)) + m\Phi_{max}} n^{\frac{-m\Phi_{max}}{1 - \Phi(\beta)\rho(G(V_p))}}\\
& = 1 - \frac{\epsilon}{\epsilon + m\Phi_{max}} n^{\frac{-m\Phi_{max}}{\epsilon}}.
\end{align*}
The probability $q$ that the infection dies before any node in the core gets infected can thus be bounded from below.
\begin{equation}
\label{eq:prob-bound}
q \geq \mathbb{P}[\zeta_1 \leq \zeta_2] \geq \left(Cn^{\frac{m\Phi_{max}}{\epsilon}}\right)^{-1},
\end{equation}
where $C = 1 + m\Phi_{max}/\epsilon>1$.

Finally, the extinction-time $\tau_e$ can now be bounded as
$$\tau_e \leq_{st} \sum_{i=1}^{V+1} U_i,$$
where $V$ is a geometric random variable with mean $\frac{1}{q}$ and $U_i$ are i.i.d. random variables independent of $V$ with mean equal to the sum of upper bounds (R.H.S) in Equations~\ref{eq:ET1-bound} and \ref{eq:ET2-bound}. Thus, using Equations~\ref{eq:prob-bound},\ref{eq:ET1-bound}, \ref{eq:ET2-bound}, we have
\begin{align*}
\mathbb{E}[\tau_e] & \leq \frac{1}{q} \EE[U_1]	\\
& \leq C n^{\frac{m\Phi_{max}}{\epsilon}}\left(\frac{(1+\Phi_{max})^m}{m \Phi_{max}} + \frac{n}{1 - \Phi(\beta)\rho(G(V_p))} \right).
\end{align*}
Simplifying we get that $\log\EE[\tau_e]= O\left(m\log n\right)$.
\end{proof}

\section{Conclusion}
We have studied the effect of bounded susceptibility on the $SI$ and $SIS$ epidemics. In particular, we showed that imposing this constraint brings out a sharp dichotomy between star-like and clique-like networks. Our results are applicable in several real-life settings which exhibit bounded susceptibility -- social networks, human-disease networks, infrastructure networks, etc. 
Along the way, we also develop a suite of lemmas relating graph structure and epidemic behavior -- some extend/simplify existing results while others are novel. We hope these will prove useful in other settings.

\section*{Acknowledgements}
This work was supported by NSF grants IIS-1017525, CNS-1320175 and ARO grant W911NF-11-1-0265.  We thank the reviewers for their suggestions which helped in greatly improving the clarity of this paper.

\begin{thebibliography}{10}
\providecommand{\url}[1]{#1}
\csname url@samestyle\endcsname
\providecommand{\newblock}{\relax}
\providecommand{\bibinfo}[2]{#2}
\providecommand{\BIBentrySTDinterwordspacing}{\spaceskip=0pt\relax}
\providecommand{\BIBentryALTinterwordstretchfactor}{4}
\providecommand{\BIBentryALTinterwordspacing}{\spaceskip=\fontdimen2\font plus
\BIBentryALTinterwordstretchfactor\fontdimen3\font minus
  \fontdimen4\font\relax}
\providecommand{\BIBforeignlanguage}[2]{{%
\expandafter\ifx\csname l@#1\endcsname\relax
\typeout{** WARNING: IEEEtran.bst: No hyphenation pattern has been}%
\typeout{** loaded for the language `#1'. Using the pattern for}%
\typeout{** the default language instead.}%
\else
\language=\csname l@#1\endcsname
\fi
#2}}
\providecommand{\BIBdecl}{\relax}
\BIBdecl

\bibitem{Kesten03}
H.~Kesten, ``First-passage percolation,'' in \emph{From classical to modern
  probability}.\hskip 1em plus 0.5em minus 0.4em\relax Springer, 2003.

\bibitem{BenPerTree}
I.~Benjamini and Y.~Peres, ``Tree-indexed random walks on groups and first
  passage percolation,'' \emph{Probability Theory and Related Fields}, 1994.

\bibitem{gopban11}
S.~Banerjee, A.~Gopalan, A.~K. Das, and S.~Shakkottai, ``Epidemic spreading
  with external agents,'' \emph{IEEE Trans. Information Theory}, 2014.

\bibitem{BhamidiGNP}
S.~Bhamidi, R.~Van~der Hofstad, and G.~Hooghiemstra, ``First passage
  percolation on the erds-renyi random graph,'' \emph{Combinatorics,
  Probability \& Computing}, 2011.

\bibitem{Liggett99}
T.~M. Liggett, \emph{Stochastic interacting systems: contact, voter and
  exclusion processes}.\hskip 1em plus 0.5em minus 0.4em\relax Springer, 1999.

\bibitem{kepwhite91:viruses}
J.~O. Kephart and S.~R. White, ``Directed-graph epidemiological models of
  computer viruses,'' in \emph{IEEE Symposium on Security and Privacy}, 1991.

\bibitem{satves02:scalefree}
R.~Pastor-Satorras and A.~Vespignani, ``{Epidemic Dynamics in Finite Size
  Scale-Free Networks},'' \emph{Phys. Rev. E}, 2002.

\bibitem{massganesh05}
A.~J. Ganesh, L.~Massouli{\'e}, and D.~F. Towsley, ``The effect of network
  topology on the spread of epidemics,'' in \emph{IEEE INFOCOM}, 2005.

\bibitem{bbcs05}
N.~Berger, C.~Borgs, J.~T. Chayes, and A.~Saberi, ``On the spread of viruses on
  the internet,'' in \emph{ACM-SIAM SODA}, 2005.

\bibitem{Banerjee2014a}
S.~Banerjee, A.~Chatterjee, and S.~Shakkottai, ``Epidemic thresholds with
  external agents,'' in \emph{IEEE INFOCOM}, 2014.

\bibitem{BorgsAntidote}
C.~Borgs, J.~Chayes, A.~Ganesh, and A.~Saberi, ``How to distribute antidote to
  control epidemics,'' \emph{Random Structures \& Algorithms}, 2010.

\bibitem{DraiefMass}
M.~Draief and L.~Massouli{\'e}, \emph{Epidemics and rumours in complex
  networks, volume 369 of London Mathematical Society Lecture Notes}.\hskip 1em
  plus 0.5em minus 0.4em\relax Cambridge University Press, Cambridge, 2010.

\bibitem{draief08}
M.~Draief, A.~Ganesh, and L.~Massoulie, ``Thresholds for virus spread on
  networks,'' \emph{The Annals of Appl. Prob.}, vol.~18, pp. 359--378, 2008.

\bibitem{Hodas12}
N.~O. Hodas and K.~Lerman, ``How visibility and divided attention constrain
  social contagion,'' in \emph{IEEE Privacy, Security, Risk and Trust}, 2012.

\bibitem{Weng12}
L.~Weng, A.~Flammini, A.~Vespignani, and F.~Menczer, ``Competition among memes
  in a world with limited attention,'' \emph{Scientific Reports}, 2012.

\bibitem{Leskovec07}
J.~Leskovec, L.~A. Adamic, and B.~A. Huberman, ``The dynamics of viral
  marketing,'' \emph{ACM Transactions on the Web (TWEB)}, 2007.

\bibitem{Funk09}
S.~Funk, E.~Gilad, C.~Watkins, and V.~A.~A. Jansen, ``The spread of awareness
  and its impact on epidemic outbreaks,'' \emph{PNAS}, 2009.

\bibitem{Wu12}
Q.~Wu, X.~Fu, M.~Small, and X.-J. Xu, ``The impact of awareness on epidemic
  spreading in networks,'' \emph{Chaos: Interdisciplinary Journal of Nonlinear
  Science}, 2012.

\bibitem{Eugster04}
P.~T. Eugster, R.~Guerraoui, A.-M. Kermarrec, and L.~Massouli{\'e}, ``From
  epidemics to distributed computing,'' \emph{IEEE Computer}, vol.~37, 2004.

\bibitem{Chakrabarti07}
D.~Chakrabarti, J.~Leskovec, C.~Faloutsos, S.~Madden, C.~Guestrin, and
  M.~Faloutsos, ``Information survival threshold in sensor and p2p networks,''
  in \emph{IEEE INFOCOM}, 2007.

\bibitem{andersonmay92:diseasesbook}
R.~M. Anderson and R.~M. May, \emph{Infectious Diseases of Humans Dynamics and
  Control}.\hskip 1em plus 0.5em minus 0.4em\relax OUP, 1992.

\bibitem{Durrett07}
R.~Durrett, \emph{Random graph dynamics}.\hskip 1em plus 0.5em minus
  0.4em\relax Cambridge university press, 2007.

\bibitem{dshah:gossipbook}
D.~Shah, ``Gossip algorithms,'' \emph{Foundations and Trends in Networking},
  2009.

\bibitem{sanghajek07:gossip}
S.~Sanghavi, B.~Hajek, and L.~Massouli{\'e}, ``Gossiping with multiple
  messages,'' in \emph{IEEE INFOCOM}, 2007.

\bibitem{janson12}
S.~Janson, T.~{\L}uczak, T.~Turova, and T.~Vallier, ``Bootstrap percolation on
  the random graph $ g\_ $\{$n, p$\}$ $,'' \emph{The Annals of Applied
  Probability}, 2012.

\bibitem{lelarge09}
M.~Lelarge, ``Efficient control of epidemics over random networks,'' in
  \emph{Proceedings of the eleventh international joint conference on
  Measurement and modeling of computer systems}, 2009.

\bibitem{chung-etal03eigenvalues-power-law}
F.~Chung, L.~Lu, and V.~Vu, ``Eigenvalues of random power law graphs,''
  \emph{Annals of Combinatorics}, 2003.

\bibitem{janson10large-clique-power-law}
S.~Janson, T.~{\L}uczak, and I.~Norros, ``Large cliques in a power-law random
  graph,'' \emph{Journal of Applied Probability}, 2010.

\bibitem{bremaud:mcbook}
P.~Bremaud, \emph{Markov Chains: Gibbs Fields, Monte Carlo Simulation, and
  Queues}.\hskip 1em plus 0.5em minus 0.4em\relax Springer-Verlag New York
  Inc., 2001.

\bibitem{jouini2008moments}
O.~Jouini and Y.~Dallery, ``Moments of first passage times in general
  birth--death processes,'' \emph{Mathematical Methods of Operations Research},
  2008.

\end{thebibliography}


\balancecolumns
\appendix
\label{sec:appendix}

In this section, we give concentration results which show that there is a sharp contrast in spread/extinction times between the bounded susceptibility and the traditional models. 

First however, we give a proof of Lemma
\ref{lemma:absorption-time} from Section~\ref{ssec:example}. This
proof is very similar to that of Lemma $8$ in \cite{Banerjee2014a} -- we
include it here mainly for the sake of completeness.

\begin{proof}[of Lemma \ref{lemma:absorption-time}]
Following \cite{Banerjee2014a}, we embed this process in another CTMC $\{\tilde{U}(t)\}$ with same transition rates \emph{except for an additional transition from $0$ to $1$ at rate $q_{0,1}$}. It is easy to see that $\{\tilde{U}(t)\}$ is now an ergodic CTMC. Recall that we defined $T(i,j) := inf\{t: U(t) = j \| U(0) = i\}$. Similarly, define $\tilde{T}(i,j) := inf\{t: \tilde{U}(t) = j \| \tilde{U}(0) = i\}$. Note that $T(1,0)$ and $\tilde{T}(1,0)$ have the same distribution, and therefore, $\EE[T(1,0)] = \EE[\tilde{T}(1,0)]$. Now, $\{\tilde{U}(t)\}$ is a finite state, irreducible CTMC and therefore, positive recurrent. Let $\pi$ be its stationary distribution. Solving for $\pi_0$ gives:
\begin{align*}
\frac{1}{\pi_0} = 1 + \sum_{k=1}^n \prod_{i=0}^{k-1} \frac{q_{i,i+1}}{q_{i+1,i}}.
\end{align*}
From standard ergodic Markov Chain theory \cite{bremaud:mcbook}, we have that $\EE[\tilde{T}(0,0)] = \frac{1/q_{0,1}}{\pi_0}$, and further, from our construction, $\EE[\tilde{T}(0,0)] = 1/q_{0,1} + \EE[\tilde{T}(1,0)].$ Thus:
\begin{align*}
\EE[\tilde{T}(1,0)] = \frac{1}{q_{0,1}}\left(\frac{1}{\pi_0}- 1\right) = \frac{1}{q_{1,0}}\left(1 + \sum_{k=2}^n \prod_{i=1}^{k-1} \frac{q_{i,i+1}}{q_{i+1,i}}\right).
\end{align*}
This completes the proof.
\end{proof}

\paragraph*{$SI$: Upper Bound in Clique-like Networks}
\begin{theorem}
\label{appthm:si-upper}
Under the conditions of Theorem \ref{thm:SI-upperbnd}, the spreading-time $\tau_s \leq \frac{(4+o(1))(\log n)^2}{\beta \eta(G)}$ with probability at least $1-\frac{1}{n}.$
\end{theorem}

\begin{proof}
Define $\Delta(k)=\frac{2\log n}{\beta\eta(G) \min (k,n-k)}$. Then we have:
\begin{align*}
\PP&\left[T_{k+1}-T_k > \Delta(k)\right]\\
& = \sum_{s_k\in\mathcal{V}_k}\PP\left[\left.T_{k+1}-T_k >\Delta(k)\right|S_k=s_k\right]\PP[S_k=s_k]\\
& = \EE\left[\exp\left(-\frac{2\beta|\delta(S_k)|\log n}{\beta\eta(G) \min (k,n-k)}\right)\right] \leq \exp(-2\log n).
\end{align*}
By telescoping, we have $\tau_s = \sum_{k=1}^{n-1} T_{k+1}-T_k$. Further: 
\begin{align*}
\sum_{k=1}^{n-1}\Delta(k)=\sum_{k=1}^{n-1} \frac{2\log n}{\beta\eta(G) \min (k,n-k)} \leq \frac{4(\log n)^2}{\beta \eta(G)},
\end{align*}
Putting all the above together, we get:
\begin{align*}
\PP\left[\tau_s > \frac{4(\log n)^2}{\beta \eta(G)}\right]&\leq \PP\left[\sum_{k=1}^{n-1} T_{k+1}-T_k > \sum_{k=1}^{n-1} \Delta(k)\right]\\
&\leq \sum_{k=1}^{n-1}\PP\left[T_{k+1}-T_k > \Delta(k)\right] \leq \frac{1}{n}.
\end{align*}
This completes the proof.
\end{proof}
Therefore, in clique-like networks with $\eta(G) = \omega\left((\log n)^2\right)$, the spread time under the $SI$ model is vanishingly small with high probability. Again, as in Proposition~\ref{prop:si-clique}, we show that this is not the case with the $SI_{bs}$ epidemic.
\paragraph*{$SI_{bs}$: Lower Bound}
\begin{theorem}
\label{appthm:si-lower}
For any graph under $SI_{bs}$ dynamics, $\tau_s > \frac{\log n}{2 \Phi_{max}}$ w.p. at least $1-\frac{1}{\sqrt{n}}.$
\end{theorem}

\begin{proof}
Since any susceptible node is infected by rate at most $\Phi_{max},$ $T_{k+1}-T_k \geq_{st} Z_k,$ where $Z_k$ has $Exponential((n-k)\Phi_{max})$ distribution. Using Chernoff bound,
\begin{align*}
\PP\left[\sum_{k=1}^{n-1} Z_k \leq t\right] & \leq e^{\Phi_{max}t} \prod_{k=1}^{n-1} \EE[e^{-\Phi_{max}Z_k}]\\
& = e^{\Phi_{max}t} \prod_{k=1}^{n-1} \frac{k}{k+1} = \frac{e^{\Phi_{max}t}}{n}.
\end{align*}
Therefore,
\begin{align*}
\PP\left[\tau_s \leq \frac{\log n}{2 \Phi_{max}}\right] \leq \PP\left[\sum_{k=1}^{n-1} Z_k  \leq \frac{\log n}{2 \Phi_{max}}\right] \leq \frac{1}{\sqrt{n}}.\\
\end{align*}
\end{proof}

\paragraph*{$SIS_{bs}$: Lower Bound in Star-like Networks}

Before we prove that the infection duration does not scale exponentially in star-like networks with high probability, we need the following lemma:
\begin{lemma}
\label{eqn:absorp-time-sq}
In a clique of size $n$, for an arbitrary initial infected node, the second moment of extinction time for $SIS_{bs}$ dynamics has an upper bound given by
$$\mathbb{E}[\tau_e^2] \leq \frac{2(1+\max(1, \Phi_{max}))^{2n}}{n \Phi_{max}^2}.$$
\end{lemma}
\begin{proof}
We again consider the evolution of the total number of infected nodes with transition rates $q_{i,i+1} \leq (n-i)\Phi_{max}$ and $q_{i+1,i} = i+1.$ The second moment for the first passage time to $0$ starting from state $1$ for a general birth and death process is derived in \cite{jouini2008moments} and is given by:
\begin{align*}
\EE[\tau_e^2] & = \frac{2}{q_{1,0}\pi_0} \sum_{i=1}^{n}\frac{1}{q_{i-1,i}\pi_{i-1}} (\sum_{k=i}^{n} \pi_k)^2	\\
& = \frac{2}{q_{1,0}\pi_0}\sum_{i=1}^{n}\frac{\pi_{i}}{q_{i,i-1}} \left(1 + \sum_{k=i}^{n-1} \prod_{l=i}^{k} \frac{q_{l,l+1}}{q_{l+1,l}}\right)^2.\\
\end{align*}
Using this equation, and $q_{i,i+1} \leq (n-i)\Phi_{max}$ and $q_{i+1,i} = i+1,$ we have
\begin{align*}
\EE[\tau_e^2] & \leq \frac{2}{n \Phi_{max}\pi_0} \sum_{i=1}^{n} \frac{\pi_0 \Phi_{max}^i \binom{n}{i}}{i} \left(\frac{\sum_{k=i}^n \Phi_{max}^k \binom{n}{k}}{\Phi_{max}^i \binom{n}{i}}\right)^2	\\
& \leq \left\{\begin{array}{lr} \frac{2}{n \Phi_{max}} \frac{(1+\Phi_{max})^{2n}}{\Phi_{max}} \quad \text{ if } \Phi_{max} > 1\\
\frac{2}{n \Phi_{max}} 2^{2n} \quad \text{ if } \Phi_{max} \leq 1 \end{array}\right.
\end{align*}
\end{proof}

\begin{theorem}
\label{appthm:sisstar}

Under the conditions of Theorem~\ref{thm:SISbs-starlike}, $\log\tau_e= O(m\log n)$ with probability at least $1-2n^{-1}-n^{\frac{-m\Phi_{max}}{2\epsilon}}$ for large enough $n.$
\end{theorem}
\begin{proof}
As in the proof of Theorem~\ref{thm:SISbs-starlike}, we consider cycles of epidemic evolution. We first bound the total number of cycles and then give a bound for the duration of each cycle. As shown in Theorem~\ref{thm:SISbs-starlike}, the total number of cycles is stochastically dominated by $V,$ a geometric random variable with mean $\frac{1}{q}.$ Let $K = Cn^{\frac{m}{\epsilon} \Phi_{max}}\log n.$ Therefore, from Equation~\ref{eq:prob-bound},
\begin{align}
\label{eqn:num-cycles-bound}
\PP[V > K] = (1-q)^{K} \leq \exp(-qK) \leq \frac{1}{n}
\end{align}
We next bound the duration of the first epoch in the cycle -- the duration of the cycle in which at least one node in the core is infected. Since this time is stochastically dominated by the duration of infection in a clique of size $m$ (size of the core), we have $$\mathbb{E}[(T_{1,j}-T_{0,j})^2] \leq \frac{2(1+\Phi_{max})^{2m}}{m \Phi_{max}^2}.$$ Using Markov inequality,
\begin{align*}
\PP[T_{1,j}-T_{0,j} > 2n^{\frac{m}{\epsilon} \Phi_{max}}] & \leq \frac{\EE[(T_{1,j}-T_{0,j})^2]}{4n^{\frac{2m}{\epsilon} \Phi_{max}}}\\
& \leq \frac{(1+\Phi_{max})^{2m}}{2m \Phi_{max}^2n^{\frac{2m}{\epsilon} \Phi_{max}}}.
\end{align*}
Using union bound, the total duration of these first epochs in $K$ cycles can bounded as follows:
\begin{align}
\label{eqn:T1-bound}
\PP[\sum_{j=1}^K T_{1,j}-T_{0,j} > 2Kn^{\frac{m}{\epsilon} \Phi_{max}}] & \leq K\frac{(1+\Phi_{max})^{2m}}{2m \Phi_{max}^2n^{\frac{2m}{\epsilon} \Phi_{max}}}	\nonumber	\\
& \leq n^{\frac{-m}{2\epsilon}\Phi_{max}}
\end{align} 
for $n$ large enough.
Similarly, we upper bound the duration of the second epoch which is stochastically dominated by $\zeta_1$ defined in Section~\ref{sec:proofs}.
\begin{align*}
\PP[T_{0,j+1}-T_{1,j} > \frac{2C}{\epsilon} \log n] \leq ne^{-2C \log n}.
\end{align*}
Therefore, the total duration of the second epochs in $K$ cycles is bounded by,
\begin{align}
\label{eqn:T2-bound}
\PP[\sum_{j=1}^K T_{0,j+1}-T_{1,j} > \frac{2KC}{\epsilon} \log n] & \leq Kne^{-2C \log n}	\nonumber\\
& = \frac{C \log n}{n^C} < \frac{1}{n}
\end{align}
for large enough $n.$ Since $\tau_e = \sum_{j=1}^V T_{0,j+1}-T_{0,j} = \sum_{j=1}^V (T_{0,j+1}-T_{1,j}) + (T_{1,j}-T_{0,j}),$ using union bound we have,
\begin{align*}
\PP[\tau_e > 2K(n^{\frac{m}{\epsilon} \Phi_{max}} + C \log n)] & \leq \PP[V > K]\\
& + \PP[\sum_{j=1}^K T_{1,j}-T_{0,j} > 2Kn^{\frac{m}{\epsilon} \Phi_{max}}]\\
& + \PP[\sum_{j=1}^K T_{0,j+1}-T_{1,j} > \frac{2KC}{\epsilon} \log n].
\end{align*}
Finally, using Equations~\ref{eqn:num-cycles-bound}, \ref{eqn:T1-bound}, \ref{eqn:T2-bound}, we have for $n$ large enough,
\begin{align*}
\PP[\tau_e > 2K(n^{\frac{m}{\epsilon} \Phi_{max}} + \frac{C}{\epsilon} \log n)] \leq 2n^{-1} + n^{\frac{-m\Phi_{max}}{2\epsilon}}.
\end{align*}
Since $\log (K(n^{\frac{m}{\epsilon} \Phi_{max}} + \frac{C}{\epsilon} \log n)) =  O(m\log n),$ we have our result.
\end{proof}

\balancecolumns
\end{document}